\documentclass[sigconf]{acmart}
\AtBeginDocument{%
  }

\setcopyright{acmlicensed}
\copyrightyear{2018}
\acmYear{2018}
\acmDOI{XXXXXXX.XXXXXXX}
\acmConference[Conference acronym 'XX]{Make sure to enter the correct
  conference title from your rights confirmation email}{June 03--05,
  2018}{Woodstock, NY}
\acmISBN{978-1-4503-XXXX-X/2018/06}

\usepackage{cleveref}
\usepackage{multirow}
\usepackage{algorithm}
\usepackage{algpseudocode}

\usepackage{subfigure}  
\usepackage{bbm}  
\usepackage{enumitem}
\usepackage{diagbox} 
\usepackage[normalem]{ulem} 

\usepackage{todonotes}

\usepackage{thmtools}  

\usepackage{makecell}

\usepackage{tcolorbox}

\newcommand{\nop}[1]{}

\newcommand{\pend}[1]{{\color{black} #1}}

\newcommand{\arev}[1]{{\color{black} #1}}

\newtheorem{remark}{Remark}

\newcommand{\figurewidthone}{41mm}  
\newcommand{\figurewidthfour}{39.5mm}  
\newcommand{\figurewidthfive}{40.1mm}  

\newcommand{\figurehreduce}{0mm}

\newcommand{\figurehreducetest}{-3.5mm}

\newcommand{\figurevreduce}{0mm}
\newcommand{\figurehreducenew}{-1.6mm}

\newcommand{\thicktoprule}{\specialrule{1.2pt}{0pt}{1.5pt}}
\newcommand{\thickmidrule}{\specialrule{1.2pt}{1.5pt}{1.5pt}}
\newcommand{\thickbottomrule}{\specialrule{1.2pt}{1.5pt}{0pt}}

\begin{document}

\title{RaMark: Radioactive Watermarking for Generated Tabular Data}



\author{Xin Che}
\affiliation{%
  \institution{McMaster University}
  \city{Hamilton}
  \country{Canada}
}
\email{chex5@mcmaster.ca}

\author{Lingyang Chu}
\affiliation{%
  \institution{McMaster University}
  \city{Hamilton}
  \country{Canada}
}
\email{chul9@mcmaster.ca}

\author{Qiqi Zhang}
\affiliation{%
  \institution{McMaster University}
  \city{Hamilton}
  \country{Canada}
}
\email{zhangq16@mcmaster.ca}

\author{Xinyu Ma}
\affiliation{%
  \institution{McMaster University}
  \city{Hamilton}
  \country{Canada}
}
\email{ma209@mcmaster.ca}

\author{Xuan Luo}
\affiliation{%
  \institution{York University}
  \city{Toronto}
  \country{Canada}
}
\email{xuanluo@yorku.ca}

\author{Jian Pei}
\affiliation{%
  \institution{Duke University}
  \city{Durham}
  \country{United States}
}
\email{j.pei@duke.edu}

\renewcommand{\shortauthors}{Anonymous}

\begin{abstract}
Recent advances in generative modeling have made generated tabular data a practical solution for privacy-sensitive data sharing, where watermarking enables ownership verification. However, existing watermarking methods fundamentally fail under retraining attacks, in which an adversary retrains a generative model on a watermarked dataset and regenerates high-utility data that no longer carries the watermark. We address this challenge by introducing \textbf{radioactivity}, the property that a watermark remains detectable after generative model retraining, and propose \textbf{RaMark}, a radioactive watermarking method that embeds a sinusoidal dependency as an intrinsic component of the data distribution. By coupling the watermark with the underlying distribution, RaMark ensures that any generative model preserving data utility also has to preserve the watermark. We theoretically show that with high probability removing watermark degrades utility and alters data distribution. Extensive experiments on two real-world tabular datasets, under a large-scale ownership verification setting with $10^5$ independent data owners, demonstrate that RaMark achieves substantially stronger radioactivity than seven state-of-the-art methods and consistently outperforms them against both retraining and data modification attacks.
\end{abstract}



\received{20 February 2007}
\received[revised]{12 March 2009}
\received[accepted]{5 June 2009}

\maketitle

\section{Introduction}
\label{sec:intro}
Recent advances~\cite{kotelnikov2023tabddpm, zhang2023mixed,du2024systematic,bauer2024comprehensive,wang2024moderator,lv2025rag,wang2024dye4ai} in generative modeling 
\pend{enable organizations} to generate and share high-quality tabular datasets instead of releasing original data.
In many domains involving sensitive information, such as healthcare~\cite{cappiello2022enabling,ceritli2023synthesizing,qian2024synthetic}, finance~\cite{sattarov2023findiff,oyewole2024data,potluru2023synthetic}, and governance~\cite{alzamil2019new,singh2020towards}, directly sharing original data is restricted by privacy regulations, compliance requirements, or contractual obligations~\cite{stoian2025survey,sattarov2023findiff,wang2024harmonic,ceritli2023synthesizing,liu2022privacy}.
To enable data sharing and trading in these scenarios, organizations increasingly rely on generative models to produce high-quality generated tabular data that approximate the original data distribution while preserving privacy~\cite{foo2025ai,qian2023synthcity,wang2025generative,lee2025removal}. These generated datasets are widely used for downstream machine learning tasks~\cite{stoian2025survey,wang2024harmonic,lee2025removal} and are often redistributed across multiple parties, making them valuable digital assets that require protection of ownership and effective traceability to identify the exact source of unauthorized redistribution.


As generated datasets are increasingly reused, redistributed, and licensed to multiple parties, dataset owners often have a strong need to establish ownership and enable 
\pend{traceability} when misuse occurs.
Data watermarking~\cite{he2024watermarking,zheng2024tabularmark,zhao2023recipe,wen2023tree} provides a principled mechanism to protect ownership by embedding a detectable secret signal into the dataset, which enables ownership verification, misuse detection, and traceability after data sharing or redistribution.

\begin{figure}[t]
    \centering
    \includegraphics[width=\linewidth]{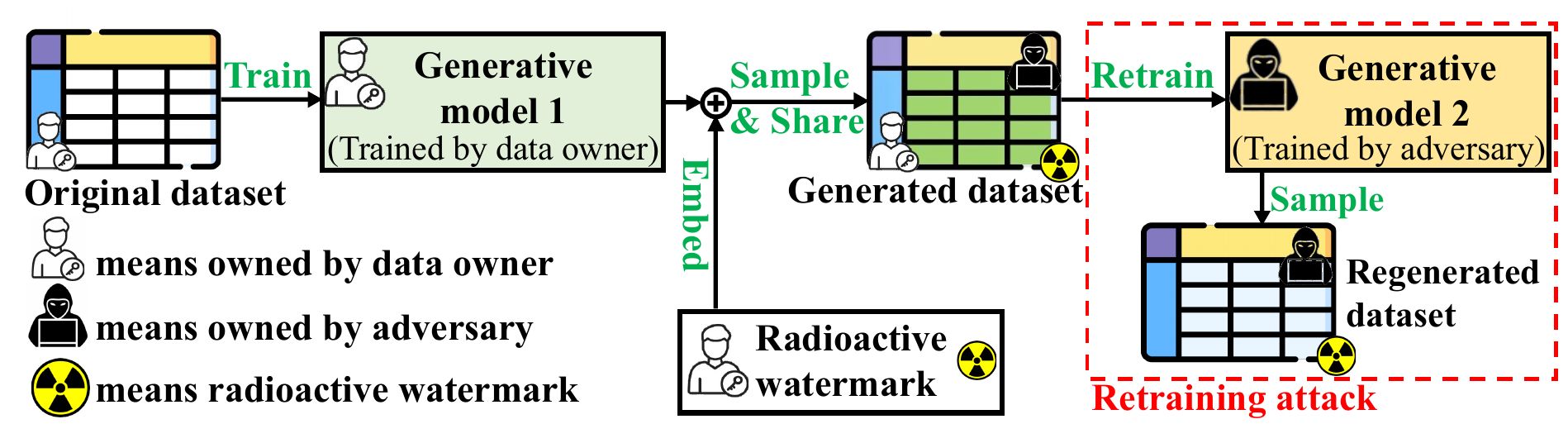}
    \caption{Illustration of retraining attack and radioactive watermark. The data owner trains a generative model on the original dataset and samples a watermarked dataset that carries a radioactive watermark. An adversary performs a retraining attack by training a new generative model on the watermarked dataset and sampling a regenerated dataset. Because the watermark is radioactive, it remains detectable in the regenerated dataset despite retraining.}
    \Description{}
    \label{fig:intro}
\end{figure}

However, the same advances in generative modeling~\cite{ho2020denoising,kindji2024under,shi2025comprehensive} that enable high-quality data sharing also equip adversaries with a powerful watermark-removing attack, namely the \textbf{retraining attack}.
As shown in \Cref{fig:intro}, an adversary retrains a generative model (Generative model~2) on a watermarked dataset and then samples a regenerated dataset from the retrained model.
Because the retrained model learns distributional structure rather than memorizing individual data samples~\cite{koo2023survey,cao2024survey,fernandez2023stable}, watermark signals embedded via record-level modifications, record identification, or generation-time heuristics are often not reproduced in regenerated data.
At the same time, the regenerated dataset can still closely approximate the original data distribution, enabling the adversary to remove the watermark without substantially degrading data utility.
This makes retraining attacks a realistic and critical threat to the watermark's security.

Retraining attacks pose a serious threat to existing watermarking methods.
Most prior approaches embed watermark signals through per-sample modifications, record-level identification, or generation-time procedures that are not captured as learnable dependencies in the data distribution.
When a generative model is retrained on a watermarked dataset, such signals are therefore not reproduced in the regenerated dataset, while the regenerated dataset still retains high utility.
This ability to remove watermarks without sacrificing much data utility exposes a fundamental weakness of existing watermarking methods for generated tabular data.

To address this weakness, we design watermark signals that \emph{persist} through generative model retraining and remain detectable in regenerated data.
Radioactive watermarks have been well established in image and model watermarking to describe signals that remain detectable after training, fine-tuning, or model reuse~\cite{sablayrolles2020radioactive,sander2024watermarking,ghali2025radioactive,li2025hmark}. 
In this paper, we extend radioactivity to the tabular data setting, where the challenge is to make the watermark survive distribution-learning and resampling.
Guided by this notion, we propose a novel watermarking method named \textbf{Radioactive Watermark (RaMark)}, which achieves strong radioactivity against retraining attacks while maintaining high robustness against many data modification attacks that attempt to remove the watermark by directly modifying the data samples in the released dataset.
We summarize our main contributions as follows.

First, we introduce watermark radioactivity for tabular datasets as a security property and propose RaMark as a novel watermarking method that achieves strong radioactivity against retraining attacks.
The key idea of RaMark is to embed a sinusoidal dependency among continuous-valued data attributes as an intrinsic component of the data distribution, so that a generative model that accurately captures the distribution also learns and reproduces the watermark in regenerated data.

Second, we theoretically analyze the radioactivity and robustness of RaMark under retraining attacks and data modification attacks.
Our analysis shows that an adversary who aims to preserve data utility needs to preserve the data distribution, which in turn preserves the embedded sinusoidal dependency and keeps the watermark detectable.
Conversely, weakening the watermark requires changing the data distribution, which is often achieved at the cost of reduced data utility.
This tradeoff makes it challenging for an adversary to remove the watermark without sacrificing data utility, leading to strong radioactivity and robustness.

Last, we conduct extensive experiments on two real-world tabular datasets to compare RaMark with seven state-of-the-art methods in terms of radioactivity and robustness under retraining and data modification attacks.
To \pend{evaluate} these methods in large-scale ownership verification scenarios, we simulate a watermark service provider serving $10^5$ independent data owners, where each owner embeds a unique watermark, and successful ownership verification must both detect the presence of a watermark and correctly identify its owner among all data owners.
The results show that RaMark achieves significantly stronger radioactivity than all baselines and consistently outperforms them under data modification attacks.


\smallskip
\noindent\textbf{Why this is a data security problem?}
Generated tabular data is increasingly used as a shareable and tradable asset in machine learning pipelines, especially in privacy-sensitive applications~\cite{koo2023survey}. Ensuring the security of such data requires protecting ownership and enabling reliable traceability under adversarial settings. Retraining attacks allow an adversary to regenerate high-utility data while removing watermark signals, which breaks existing ownership protection mechanisms. Therefore, watermarking for generated tabular data is fundamentally a data security problem, where the goal is to ensure robustness of watermark signals against adversarial attempts under realistic reuse and regeneration workflows.

\smallskip
\noindent\textbf{Paper organization.}
\Cref{sec:related_work} reviews related watermarking methods and explains why they fail under retraining attacks.
\Cref{sec:methods} presents RaMark, including the watermark detection and embedding procedures.
\Cref{sec:robustness} analyzes the robustness and radioactivity of RaMark under retraining and data modification attacks.
\Cref{sec:exp} reports the experimental evaluation and parameter analysis.
Finally, \Cref{sec:conclusions} concludes the paper and discusses future directions.

\section{Related Work}
\label{sec:related_work}

Existing research on generated tabular data watermarking can be broadly divided into two categories:
(1) \textbf{generative methods}~\cite{zhu2025tabwak,chen2025tag}, which modify the data generation process to embed watermarks, and
(2) \textbf{database watermarking methods}~\cite{li2005fingerprinting,sebe2006watermarking,shehab2007watermarking,imamoglu2017new,hu2018new,liu2022random,wang2023copyright,kamran2018comprehensive}, which embed watermarks into data after generation.
A common limitation of most existing methods is that their watermark signals are not embedded as stable, learnable dependencies of the data distribution.
Since retraining attacks operate by learning and reproducing the distribution rather than memorizing individual records, watermark signals that remain external to the distribution are typically not reproduced in regenerated data.


In the following, we review representative generative and database watermarking methods and analyze why they do not achieve \emph{radioactivity}, i.e., why their watermark signals are not preserved under retraining attacks. Additional discussion on prior studies of radioactivity is deferred to Appendix~\ref{app:prior_radioactivity}.


\emph{Generative watermarking methods}
embed watermarks during the data generation process.
TabWak~\cite{zhu2025tabwak} perturbs the latent variables of a diffusion model to generate watermarked samples.
However, due to the highly nonlinear mapping between latent and data spaces, identical latent perturbations can manifest as inconsistent, sample-specific distortions in data space.
These distortions do not induce coherent statistical dependencies in the overall data distribution.
As a result, when a generative model is retrained on the watermarked dataset, the watermark signal is not reliably reproduced, and TabWak does not achieve radioactivity.

MUSE~\cite{fang2025muse} embeds watermarks by using a scoring function to select specific generated samples.
This mechanism introduces per-record selection bias, but it does not establish global structural dependencies among attributes.
Consequently, the watermark signal is tied to a generation-time selection rule rather than to the underlying data distribution.
A retrained generative model that learns the distribution of the dataset does not inherit this selection rule, and the watermark is therefore not preserved.




Overall, existing generative methods rely on model-specific perturbations or procedural selection rules that are not encoded as intrinsic distributional structure.
Such signals are fragile under retraining attacks.


\emph{Database watermarking methods} embed watermarks after data generation.
A representative class~\cite{li2005fingerprinting,sebe2006watermarking,shehab2007watermarking,imamoglu2017new,hu2018new,liu2022random,wang2023copyright,kamran2018comprehensive} modifies selected data samples identified through a primary key~\cite{agrawal2002watermarking,agrawal2003watermarking,sion2003rights,li2005fingerprinting,gupta2009reversible,gupta2009database,wang2023copyright} or a virtual primary key~\cite{shehab2007watermarking,gort2019hqr,gort2020double,zheng2024tabularmark}.
Watermark detection depends on re-identifying the same records through these keys.
Under retraining attacks, however, new samples are generated from the learned distribution rather than reproducing original records.
The correspondence required for key-based re-identification is therefore lost, making such methods inherently non-radioactive.

Another line of work avoids explicit keys and embeds watermarks by modifying statistical properties of the dataset~\cite{che2025primary,he2024watermarking,ngo2024adaptive,zheng2025b2mark}.
PKF~\cite{che2025primary} introduces artificial correlations among attributes.
However, these correlations are externally imposed perturbations rather than stable dependencies of the original distribution, and they are not reliably captured by retrained generative models.
WGTD~\cite{he2024watermarking}, NgoMark~\cite{ngo2024adaptive}, and B2Mark~\cite{zheng2025b2mark} adjust the frequencies of selected attribute values using the ``green red list'' mechanism~\cite{guo2024context,kirchenbauer2023watermark,zhao2023provable}.
These methods rely on sample-level frequency bias and do not establish coherent inter-attribute structure.
When a generative model is retrained, such localized perturbations are typically smoothed out, and the watermark signal is not preserved.

Overall, whether based on key-dependent record modification or key-free statistical perturbation, existing database watermarking methods embed signals that remain external to the data-generating structure.
They are therefore vulnerable to retraining attacks.
\section{The RaMark Method}
\label{sec:methods}
In this section, we present RaMark, a radioactive watermarking method for continuous-valued tabular data. 
RaMark is built on a \textbf{central principle}: instead of perturbing individual data samples, it embeds a sinusoidal dependency among attributes as an intrinsic component of the data distribution. 
This sinusoidal 
dependency 
\pend{serves as} the \textbf{watermark signal} carried by the dataset.
Consequently, any generative model retrained on the watermarked dataset necessarily reproduces this sinusoidal dependency when approximating the data distribution. 

\pend{This design leads to the \textbf{key intuition} behind RaMark: preserving the watermarked data distribution also preserves the dependency that carries the watermark. 
\Cref{sec:theory} formalizes this intuition by analyzing how distributional closeness affects the spectral evidence used for detection.}
\pend{In particular, our analysis shows that a more accurate approximation of the data distribution leads to stronger persistence of the watermark signal in regenerated data, which underlies the radioactivity of RaMark.}

RaMark has \textbf{two modules}: watermark detection and watermark embedding. 
\pend{The \emph{detection module} first maps a dataset into a secret two-dimensional projected space and converts the projected points into a discrete-time signal. It then checks whether the signal contains a strong sinusoidal component at the designated frequency. }
\pend{The \emph{embedding module} embeds this sinusoidal structure through watermark-guided diffusion sampling.}

\pend{We present the detection module first, because it defines the structure of the watermark signal, and then describe how the embedding module embeds this structure during data generation.}

\subsection{Watermark Detection}
\label{sec:wm_detect}

\begin{figure}[t]
    \centering
    \includegraphics[width=1\linewidth]{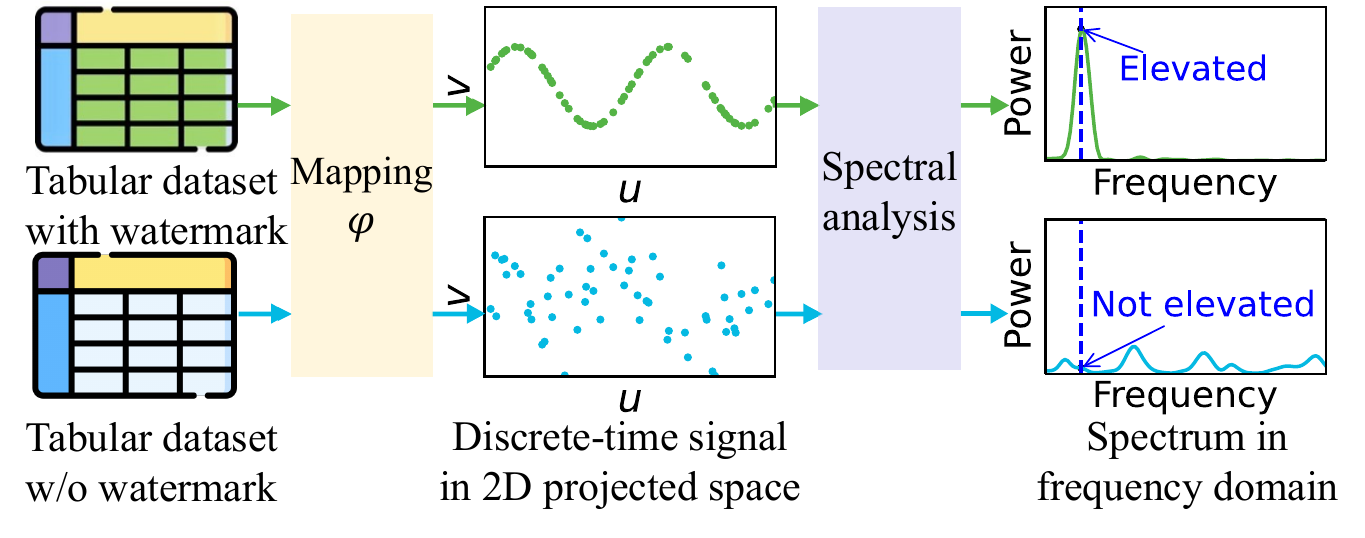}
    \caption{Watermarked and unwatermarked signals in the two-dimensional projected space and their spectra.}
    \Description{}
    \label{fig:signal_appearance}
\end{figure}

The watermark signal embedded by RaMark manifests as sinusoidal dependency of a designated frequency in a secret two-dimensional projected space. 
Watermark detection therefore 
\pend{reduces} to verify whether the projected data exhibits high spectral power at the designated frequency.

As shown in \Cref{fig:signal_appearance}, to detect the watermark, we transform the tabular dataset into a discrete-time signal in the projected space and analyze its spectral power. 
If the spectral power at the designated frequency shows a significantly elevated peak, this indicates the presence of the watermark. 
The detection procedure consists of \textbf{two steps}: (i) mapping the tabular dataset to a discrete-time signal in the two-dimensional projected space, and (ii) measuring the spectral power of the signal at the designated frequency.


%
%

\subsubsection{Mapping a tabular dataset to a discrete-time signal}
\label{sec:mapping}

To detect a sinusoidal dependency embedded in the data distribution, we first map the dataset into a two-dimensional projected space. 
Within this projected space, one coordinate induces an ordered indexing structure, while the other coordinate captures the corresponding projected values. 
By organizing samples according to this indexing structure, the sinusoidal dependency that constitutes the watermark signal manifests as a discrete-time signal, which is later processed by spectral analysis for watermark detection.



%
%
%
%
Let $Q \in \mathbb{R}^{n \times d}$ denote a \textbf{table} (i.e., tabular dataset) with $n$ rows and $d$ columns, where $Q_{i,:}$ is the $i$-th data sample. Each sample is mapped to a \textbf{point} $(u_i, v_i)$ in the two-dimensional \textbf{projected space} by projection onto two orthogonal unit vectors $\mathbf{e}_u, \mathbf{e}_v \in \mathbb{R}^d$. Specifically,
\begin{equation}
\label{eq:mapping_y}
v_i = \phi_v(Q_{i,:}) = Q_{i,:}\mathbf{e}_v^\top,
\end{equation}
and
\begin{equation}
\label{eq:mapping_x}
u_i = \phi_u(Q_{i,:}) = \frac{\left\lfloor \frac{Q_{i,:}\mathbf{e}_u^\top}{\beta} \right\rfloor}{s},
\end{equation}
\pend{where $\beta \in \mathbb{R}^+$ is the bin width that determines the granularity of the binning operation along the $u$-axis, $s \in \mathbb{R}^+$ is a scaling factor that rescales the spacing of the resulting discrete-time signal along the $u$-axis,}
and $\lfloor\cdot\rfloor$ denotes the flooring operator to perform binning. 
Mapping all the data samples in $Q$ yields a set of points
$Z = \{(u_i, v_i) \mid i=1,\ldots,n\}$.
Points with the same $u$-coordinate are grouped into bins
\begin{equation}
B_c = \{(u_i, v_i) \mid u_i = c\}.
\end{equation}
For each bin, we compute the \textbf{mean point} $(\bar{u}, \bar{v})$, where $\bar{u}=c$ and $\bar{v}$ is the mean of the corresponding $v_i$ values. 
The binning operation in \Cref{eq:mapping_x} quantizes the projection along $\mathbf{e}_u$ into an ordering index, the corresponding $v$-values serve as signal values, and the sequence of mean points forms a \textbf{discrete-time signal} $S=\varphi(Q)$.
\Cref{fig:points_and_mean_points} illustrates this aggregation: \pend{projected points are grouped into bins and converted into mean points, which form an ordered sequence used as the signal for detection.}

\begin{figure}[t]
    \centering
    \includegraphics[width=0.93\linewidth]{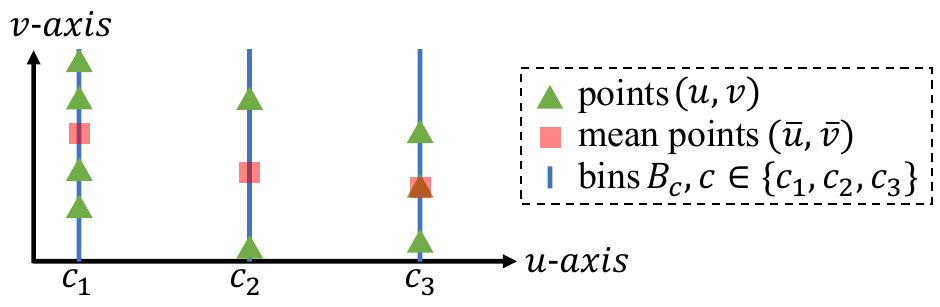}
    \caption{An example of points and mean points. The vertical lines represent the bins $B_c$ for the values of $c\in\{c_1, c_2, c_3\}$.}
    \Description{}
    \label{fig:points_and_mean_points}
\end{figure}

\begin{algorithm}[t]
  \caption{Watermark detection}
  \begin{algorithmic}[1]
    \Require A suspicious dataset $Q$ and a secret key $\mathcal{K}= \{\mathbf{e}_u, \mathbf{e}_v, \beta, s, \omega\}$.
    \Ensure Detection score $\mathtt{DS}(\omega)$.
    \State Obtain the discrete-time signal $S = \varphi(Q)$.
    \State Compute the spectral power $L(\omega)$ of $S$.
    \State Compute $\mathtt{FAP}(\omega)$ based on $L(\omega)$.
    \State Return the detection score $\mathtt{DS}(\omega)=1 - \mathtt{FAP}(\omega)$.
  \end{algorithmic}
  \label{alg:detect}
\end{algorithm}

\subsubsection{Spectral analysis and detection}
\label{sec:sawd}

Given the discrete-time signal $S$, watermark detection evaluates whether the signal exhibits high spectral power at the designated frequency $\omega$ corresponding to the sinusoidal dependency. 
If such high spectral power at the designated frequency is observed, it indicates that the projected dataset contains the sinusoidal dependency that forms the watermark signal, which implies the presence of watermark.

%
%
%
%

\pend{Because some bins may be empty, the resulting signal is not necessarily uniformly sampled. We therefore use the Lomb--Scargle Periodogram (LSP)~\cite{lomb1976least}, which estimates spectral power for non-uniformly sampled signals.}



Let $L(\omega)$ denote the spectral power of $S$ at frequency $\omega$. We define the \textbf{detection score} as
\begin{equation}
\mathtt{DS}(\omega) = 1 - \mathtt{FAP}(\omega),
\end{equation}
where $\mathtt{FAP}(\omega) \in [0,1]$ is the false alarm probability computed by LSP, and it represents the probability that random noise would produce a spectral power of at least $L(\omega)$. 
A large $\mathtt{DS}(\omega)$, thus a small $\mathtt{FAP}(\omega)$,
indicates that the spectral power at $\omega$ is unlikely to arise from noise. This indicates a strong watermark signal, which is strong evidence for the presence of the watermark. 
The final decision on the presence of a watermark is obtained by comparing $\mathtt{DS}(\omega)$ with a predefined threshold.

Algorithm~\ref{alg:detect} summarizes the detection procedure for determining whether a suspicious dataset $Q$ contains the watermark signal specified by the owner's secret key.
The \textbf{secret key}, denoted by $\mathcal{K}=\{\mathbf{e}_u,\mathbf{e}_v,\beta,s,\omega\}$, specifies the projection directions, binning parameters, and designated frequency. The key is known only to the data owner and prevents adversaries from identifying or weakening the watermark signal. The detection algorithm runs in $O(n(d+1))$ time and uses $O(n(d+1)+d)$ space.

The detection module of RaMark is inspired by PKF \cite{che2025primary} because both methods detect watermark by examining a sinusoidal signal in the 
\pend{projected} space. However, RaMark differs fundamentally from PKF because PKF's watermark is not radioactive, while RaMark's watermark is highly radioactive. Specifically, the watermark of PKF imposes artificial correlations among attributes, and these externally added correlations do not form stable or learnable dependencies in the underlying data distribution, so retrained generative models do not reproduce the watermark.

\subsection{Watermark Embedding}
\label{sec:embedding}

To embed a sinusoidal dependency into generated data, RaMark performs watermark-guided diffusion sampling, which is driven by an analytic watermark likelihood that is compatible with the detection procedure. 
Rather than modifying generated samples in a post-hoc manner, the watermark is embedded during the reverse denoising process (i.e., diffusion sampling) of a diffusion model, so that the sampled dataset intrinsically conforms to the target sinusoidal dependency in the projected space. 


We first provide an overview of the key idea, then define the watermark likelihood, and finally present watermark-guided diffusion sampling and the complete embedding algorithm.

%
%
%



%

\begin{figure}[t]
    \centering
    \includegraphics[width=\linewidth] {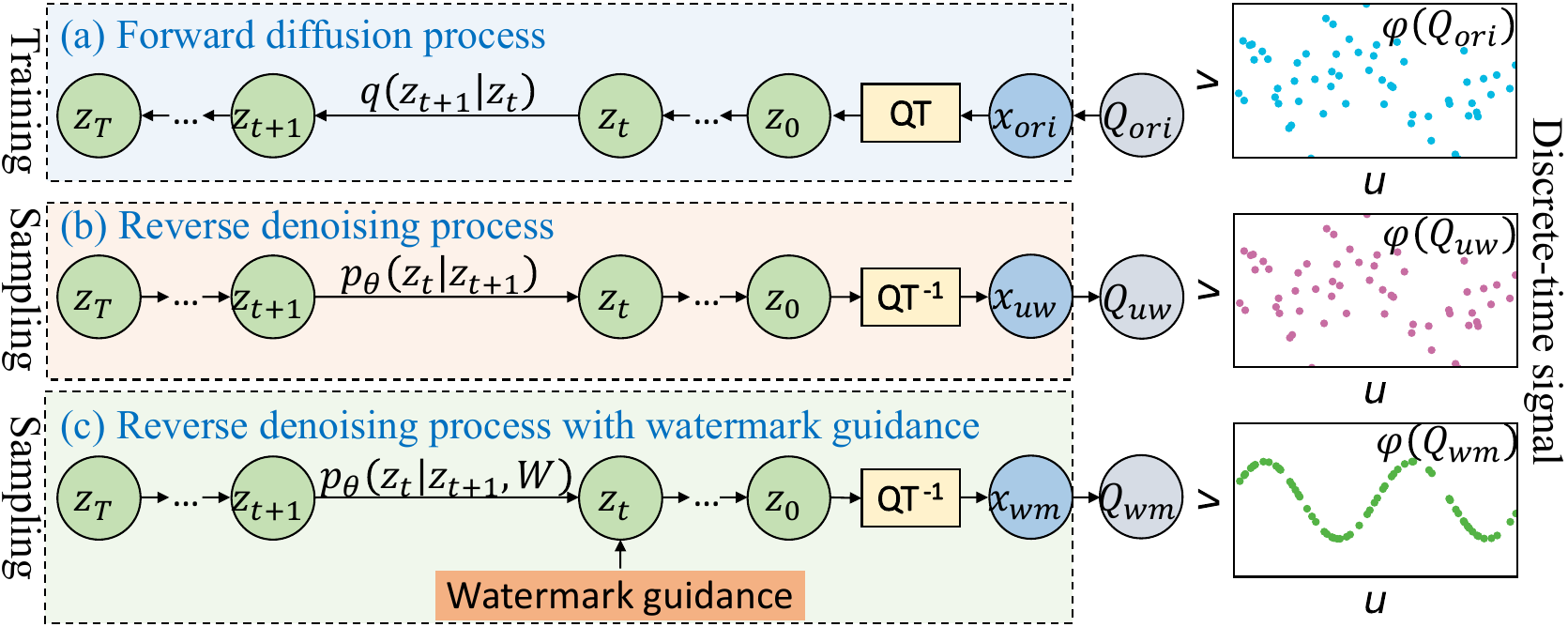}
    \caption{Overview of the watermark embedding method. The three scattered plots on the right side show the discrete-time signals of the corresponding tabular datasets $Q_{ori}$, $Q_{uw}$, and $Q_{wm}$, respectively.}
    \Description{}
    \label{fig:method_framework}
\end{figure}

\subsubsection{Overview of key idea.}
\label{sec:overview_key_idea}

\Cref{fig:method_framework} 
\pend{illustrates the embedding pipeline by contrasting standard diffusion sampling with RaMark's watermark-guided sampling.}
In \Cref{fig:method_framework}(a), the data owner trains a diffusion model on the original (unwatermarked) dataset $Q_{ori}$. In \Cref{fig:method_framework}(b), standard diffusion sampling from the trained diffusion model produces an unwatermarked generated dataset $Q_{uw}$, whose mapped discrete-time signal does not exhibit the sinusoidal dependency. In \Cref{fig:method_framework}(c), we add watermark guidance to modify the standard diffusion sampling so that the sampled dataset $Q_{wm}$ is biased to exhibit the sinusoidal dependency, thereby successfully embedding the watermark.


The quantile transformation ($\mathtt{QT}$) in \Cref{fig:method_framework}(a) is a standard normalization step in diffusion-based tabular generation~\cite{jia2024tabular,kotelnikov2023tabddpm,sattarov2023findiff}. It maps each data sample $x_{ori}\in Q_{ori}$ in the original $x$-space to a normalized latent state $z_0=\mathtt{QT}(x_{ori})$ in the $z$-space, which is the starting state of the forward diffusion process. Because $\mathtt{QT}$ is applied element-wise on attributes, the $x$-space and $z$-space have the same dimensionality~\cite{bogner2012normal,peng2007normalizing}. The training of the diffusion model applies the forward diffusion process to latent states $z_0,\ldots,z_T$ to learn the reverse denoising transitions. During sampling in \Cref{fig:method_framework}(b)--(c), we initialize $z_T\sim \mathcal{N}(0,\mathbf{I})$, iteratively denoise in the $z$-space to obtain $z_0$, and decode $z_0$ back to the $x$-space by $x=\mathtt{QT}^{-1}(z_0)$.

\subsubsection{Watermark likelihood.}

The watermark likelihood quantifies the degree to which a latent state $z_t$ aligns with the sinusoidal dependency in the projected space. 
To guide diffusion sampling toward watermarked samples, we need a differentiable objective that quantifies how well a candidate sample matches the target sinusoidal dependency. The watermark likelihood serves exactly this purpose.

During diffusion sampling, let $z_t$ denote the latent state at timestep $t$, and let $x_t = \mathtt{QT}^{-1}(z_t)$ be the corresponding decoded data sample. 
We map $x_t$ into the same two-dimensional projected space used in watermark detection. This obtains $(u_t, v_t)$, where $u_t = \phi_u(x_t)$ and $v_t = \phi_v(x_t)$. 
\pend{A sample is more conformed to the watermark if its projected value $v_t$ lies closer to the target sinusoidal curve at position $u_t$.}
The \textbf{deviation} from the sinusoidal dependency, modelled as $v = \sin(2\pi \omega u)$, is defined as
\begin{equation}
\label{eq:dxt}
d(x_t) = \left| v_t - \sin(2\pi \omega u_t) \right|,
\end{equation}
and the \textbf{watermark likelihood} is defined as
\begin{equation}
\label{eq:pwx}
\Pr(W \mid z_t) \triangleq \exp\!\left(-d(x_t)\right),
\end{equation}
where $W$ denotes the event that the decoded data sample conforms to the sinusoidal dependency. 
A small deviation $d(x_t)$ corresponds to a strong conformity with the sinusoidal dependency and therefore yields a higher watermark likelihood.

\subsubsection{Watermark-guided sampling.}
\label{sec:wm_guidance_embed}
Watermark-guided sampling modifies the standard diffusion sampling by incorporating the watermark likelihood as a multiplicative weighting term. 
The \textbf{key idea} of watermark-guided sampling is that among all denoising directions proposed by the diffusion model, we prefer those that make the sample more consistent with the target sinusoidal dependency.
At timestep $t$, we interpret the guided reverse denoising step as sampling from the conditional transition
\begin{equation}
\label{eq:transition_wm}
p_\theta(z_t \mid z_{t+1}, W) \propto p_\theta(z_t \mid z_{t+1}) \Pr(W \mid z_t),
\end{equation}
where $p_\theta(z_t \mid z_{t+1})$ formulates the standard diffusion sampling learned from $Q_{ori}$ in \Cref{fig:method_framework}(a), and $\theta$ denotes the trained diffusion model parameters. This modifies the sampling distribution of $z_t$ so that the latent states, which decode to samples with reduced deviation from the sinusoidal dependency, receive a higher probability of being sampled.

In diffusion models, the diffusion sampling is parameterized as a Gaussian
\begin{equation}
\label{eq:rev_transit_embed}
p_\theta(z_t\mid z_{t+1})=\mathcal{N}\big(z_t;\mu,\Sigma\big),
\end{equation}
where $\mu=\mu_\theta(z_{t+1},t)$ and $\Sigma=\Sigma_\theta(z_{t+1},t)$ are predicted by the model~\cite{yang2023diffusion,ho2020denoising}. 
To obtain a tractable sampling rule for \Cref{eq:transition_wm}, we locally linearize $\log \Pr(W\mid z_t)$ around $z_t=\mu$. 
The following result shows that incorporating the watermark likelihood modifies the reverse transition through a simple mean shift proportional to the gradient of the log watermark likelihood.

\begin{restatable}{theorem}{NameA}
\label{thm:gaussian_appr_embed}
Given $p_\theta(z_t\mid z_{t+1})=\mathcal{N}\big(z_t;\mu,\Sigma\big)$, the reweighted density
$p_\theta(z_t\mid z_{t+1})\Pr(W\mid z_t)$
is approximated by
\[
    p_\theta(z_t\mid z_{t+1}) \Pr(W\mid z_t)
    \approx
    \mathcal{N}\big(z_t;\mu+\Sigma \cdot g_z,\Sigma\big),
\]
where
$g_z = \nabla_{z_t}\log \Pr(W\mid z_t)\big|_{z_t=\mu}$ is the gradient of $\log \Pr(W\mid z_t)$ at $z_t=\mu$, and $\Sigma \cdot g_z$ denotes matrix-vector multiplication.
\end{restatable}

\noindent
The proof of \Cref{thm:gaussian_appr_embed} is provided in the Appendix~\ref{app:the_proof}.

By \Cref{thm:gaussian_appr_embed}, the guided reverse denoising step replaces the mean $\mu$ with $\mu+\Sigma\cdot g_z$, where $g_z$ points in the direction that increases $\log \Pr(W\mid z_t)$. Since $\log \Pr(W\mid z_t)=-d(x_t)$, this guidance tends to reduce $d(x_t)$ and therefore strengthens the embedded watermark signal. Repeating the guided mean shift over all timesteps progressively biases sampling toward latent states whose decoded data samples are more conformed to the sinusoidal dependency. This yields a dataset carrying the embedded watermark signal.

We summarize the complete embedding procedure in \Cref{alg:embed}. Repeating \Cref{alg:embed} independently produces $Q_{wm}$. Following~\citet{dhariwal2021diffusion}, we include a guidance strength parameter $\alpha>0$ that controls the influence of the watermark term:
\begin{equation}
\label{eq:final_zt}
z_t \sim \mathcal{N}\big(\mu+\alpha(\Sigma \cdot g_z),\Sigma\big).
\end{equation}
\pend{
This design introduces only controlled and minor perturbations to the generated samples, since the watermark guidance is applied gradually at each denoising step and only slightly adjusts the generated values. The guidance strength $\alpha$ further controls how strongly diffusion sampling is biased toward the target sinusoidal dependency. Empirical studies on the effect of $\alpha$ are provided in \Cref{sec:exp_params_ans}.
}

\subsubsection{Scope of application}
RaMark relies on diffusion models for watermark embedding, as the watermark is embedded during the reverse denoising process through watermark-guided sampling. We adopt diffusion models because they provide strong distribution modeling capability and represent the state-of-the-art for tabular data generation~\cite{koo2023survey,kindji2024under,zhao2023recipe}.
RaMark operates on continuous-valued attributes, which are widely present in real-world tabular datasets such as financial, medical, industrial, and sensor data. 
In mixed-type tables containing both continuous-valued and categorical attributes, the sinusoidal dependency can be embedded on any subset of continuous columns without modifying discrete or categorical fields. 
As long as at least two continuous attributes are available, RaMark can construct the required two-dimensional projected space and embed the sinusoidal dependency while preserving the remaining schema. 
This design allows the watermark signal to be integrated into a broad range of practical tabular datasets without altering their discrete components.

\begin{algorithm}[t]
  \caption{Watermark embedding}
  \begin{algorithmic}[1]
    \Require A secret key $\mathcal{K}$ and a trained diffusion model $\theta$.
    \Ensure A watermarked data sample $x_{wm}$.
    \State Sample $z_T$ from $\mathcal{N}(0,\mathbf{I})$
    \For{each $t \in \{T-1, T-2, \ldots, 1, 0\}$}
	\State Compute $\mu=\mu_\theta(z_{t+1},t)$, $\Sigma=\Sigma_\theta(z_{t+1},t)$ and $g_z$
        \State Sample $z_t$ from $\mathcal{N}\big(\mu + \alpha(\Sigma \cdot g_{z}),\Sigma\big)$
    \EndFor
    \State \Return $x_{wm}\leftarrow \mathtt{QT}^{-1}(z_0)$
  \end{algorithmic}
  \label{alg:embed}
\end{algorithm}

\section{Robustness of the Radioactive Watermark}
\label{sec:robustness}

Robustness to watermark removal attacks is a central criterion for evaluating watermarking methods. In this section, we formalize the adversary’s goal and threat model, describe representative watermark removal attacks, and then provide a theoretical analysis of the robustness and radioactivity of RaMark.

\subsection{Adversary's Goal and Threat Model}
\label{sec:threat_model}

In tabular data watermarking~\cite{kamran2018comprehensive,ren2023robust,he2024watermarking,zheng2024tabularmark,li2022secure,zhu2025tabwak}, an adversary aims to weaken or remove the watermark while preserving the utility of the released dataset. We formalize the adversary’s goal and capabilities as follows.

\subsubsection{Adversary’s goal.}
Given a watermarked dataset $Q_{wm}$, the adversary seeks to construct an attacked dataset $Q_{atk}$ that minimizes the watermark detection score without substantially degrading its data utility. 


We measure data utility by how well $Q_{atk}$ preserves the statistical properties required for downstream learning tasks. Following established works~\cite{kotelnikov2023tabddpm,shi2024tabdiff,zhu2025tabwak}, we quantify utility using \textbf{Machine Learning Efficiency} ($\mathtt{MLE}$)~\cite{al2015efficient,kotelnikov2023tabddpm,shi2024tabdiff,zhu2025tabwak}, which evaluates how well a model trained on $Q_{atk}$ generalizes to real data. The computation of $\mathtt{MLE}$ is detailed in \Cref{sec:exp_set}. Under this objective, an attack is considered successful only if it weakens or removes the watermark signal while preserving $\mathtt{MLE}$.

\subsubsection{Threat model (adversary's capabilities)}
We adopt a standard private-key watermarking setting~\cite{kamran2018comprehensive,ren2023robust,he2024watermarking,zheng2024tabularmark,li2022secure,zhu2025tabwak} with the following assumptions.
\begin{itemize}
    \item \emph{Access to released data.} The adversary has full access to the released watermarked dataset $Q_{wm}$, but not to the original dataset $Q_{ori}$, which remains private to the data owner.

    \item \emph{No access to the secret key.} The adversary does not possess the secret key $\mathcal{K}$, which specifies the projection directions and watermark parameters. The key is retained solely by the data owner~\cite{che2025primary,zheng2024tabularmark,kamran2018comprehensive}.

    \item \emph{Semantic consistency of attributes.} The adversary cannot alter the semantic meaning of attributes~\cite{che2025primary,zheng2025b2mark}. Changing attribute semantics would produce a dataset that is incompatible with the original task and would severely reduce $\mathtt{MLE}$, which violates the adversary’s goal of preserving utility.
\end{itemize}

Under this threat model, the adversary is allowed to modify data values or retrain generative models, but must preserve distributional characteristics sufficient to keep the data's utility for downstream learning tasks. This captures realistic misuse scenarios in data sharing settings, where redistributed datasets are expected to remain functionally equivalent to the original release.

\subsection{Watermark Removal Attacks}
\label{sec:wmr_attacks}

Watermark removal attacks attempt to transform a watermarked dataset $Q_{wm}$ into an attacked dataset $Q_{atk}$ in which the watermark signal is weakened or removed while maintaining data utility. We consider two broad categories of attacks, distinguished by whether the adversary regenerates the data distribution or directly edits the released dataset.

\subsubsection{Retraining attack.}
In a retraining attack, the adversary trains a new generative model on the watermarked dataset $Q_{wm}$ and then samples a regenerated dataset $Q_{atk}$ from the newly trained model. The underlying intuition is that a well-trained generative model captures the overall data distribution of $Q_{wm}$ but may fail to reproduce watermark signals that are not encoded as stable distributional dependencies. If successful, this attack produces $Q_{atk}$ that retains high data utility while weakening the watermark signal so that it becomes harder to detect the watermark.


The adversary may employ any generative model suitable for tabular data. In our experiments, we consider three major classes of deep generative models:  
(1) \emph{Diffusion models}: TabDDPM~\cite{kotelnikov2023tabddpm} and TabDiff~\cite{shi2024tabdiff};  
(2) \emph{Generative adversarial networks (GANs)}: CTAB-GAN+~\cite{zhao2021ctab,zhao2024ctab}; and  
(3) \emph{Variational autoencoders (VAEs)}: TVAE~\cite{xu2019modeling}.  
These models represent state-of-the-art generative paradigms for tabular data and provide strong retraining baselines.

\subsubsection{Data modification attacks.}
In contrast to retraining, data modification attacks directly edit the released dataset $Q_{wm}$ through value-level or structure-level operations. Such transformations are standard in the dataset watermarking literature and aim to weaken the watermark by perturbing the data while preserving its overall statistical characteristics.

We consider the following common operations:

\begin{enumerate}
\item \emph{Noise addition}~\cite{zheng2025b2mark,che2025primary,he2024watermarking}:  
Each attribute value in $Q_{wm}$ is perturbed by adding noise sampled from a uniform distribution $U[-\rho_{na}, \rho_{na}]$, where $\rho_{na} > 0$ controls the perturbation magnitude.

\item \emph{Row deletion}~\cite{kamran2018comprehensive,ren2023robust}:  
A fraction $\rho_{rd} \in [0,1]$ of rows is randomly removed from $Q_{wm}$.

\item \emph{Row insertion}~\cite{shehab2007watermarking,hu2018new,kamran2018comprehensive}:  
A fraction $\rho_{ri} \in [0,1]$ of new rows is inserted into $Q_{wm}$. For each inserted row, the $j$-th attribute value is sampled from $U[\mu_j - \sigma_j, \mu_j + \sigma_j]$, where $\mu_j$ and $\sigma_j$ are the mean and standard deviation of the $j$-th attribute in $Q_{wm}$.

\item \emph{Column deletion}~\cite{che2025primary,zheng2025b2mark}:  
A fraction $\rho_{cd} \in [0,0.5]$ of attributes (i.e., columns) is removed from $Q_{wm}$.
\end{enumerate}

For all attacks, larger values of $\rho_{na}$, $\rho_{rd}$, $\rho_{ri}$, and $\rho_{cd}$ correspond to stronger attack strength. These operations model realistic adversarial behaviors that attempt to weaken or remove watermark signal while preserving data utility for downstream tasks.

\subsection{Robustness Analysis}
\label{sec:theory}

This subsection links \emph{distributional similarity} between two tabular datasets to the \emph{spectral power} of their mapped discrete-time signals at the designated frequency $\omega$.
We first present two theorems: the first converts closeness between the signal distributions into closeness of spectral power, and the second shows that closeness between dataset distributions is (with high probability) preserved by the dataset-to-signal mapping $\varphi$.
We then give three remarks that explain why these results imply strong radioactivity and robustness of RaMark against watermark removal attacks.
The proofs of both theorems are provided in the Appendix~\ref{app:the_proof}.

\subsubsection{Setup}
Let $Q_A$ and $Q_B$ be two tabular datasets, each containing $n$ samples, with corresponding discrete-time signals
$S_A=\{(\bar{u}_j^A,\bar{v}_j^A)\}_{j=1}^{\ell}$ and
$S_B=\{(\bar{u}_j^B,\bar{v}_j^B)\}_{j=1}^{\ell}$ obtained by applying the same mapping $\varphi$ (hence the same secret key).
Let
$V=\max_{1\le j\le \ell}\{|\bar v_j^A|,|\bar v_j^B|\}$ be an amplitude bound, and let
$\sigma_A^2$ and $\sigma_B^2$, respectively, be the sample variances of $\{\bar v_j^A\}_{j=1}^{\ell}$ and $\{\bar v_j^B\}_{j=1}^{\ell}$, with $\sigma^2=\min\{\sigma_A^2,\sigma_B^2\}$.
Denote by $\overline P_A$ and $\overline P_B$ the distributions of $S_A$ and $S_B$, and by $L_A(\omega)$ and $L_B(\omega)$ their spectral powers at frequency $\omega$, resepctively.
Denote by $\mathcal{W}_1(\cdot,\cdot)$ the Wasserstein-1 distance between two distributions.

\begin{restatable}[Probabilistic Bound on Spectral Power Difference]{theorem}{maintheoremA}
\label{thm:ls-power-bound-bar}
For any $\epsilon>0$,
\[
\Pr\bigl(|L_A(\omega)-L_B(\omega)|\ge \epsilon\bigr)
\le
4\exp\Big(-\frac{\ell\epsilon^2}{C_1}\Big)
+ \frac{C_2\mathcal{W}_1(\overline{P}_A,\overline{P}_B)}{\epsilon},
\]
where $C_1=\frac{288V^4(1+V^2/\sigma^2)^2}{\sigma^4}$ and $C_2=\frac{12V(V^2+\sigma^2)}{\sigma^4}$.
\end{restatable}

\Cref{thm:ls-power-bound-bar} states a \textbf{clean message}: \emph{if the signal distributions $\overline P_A$ and $\overline P_B$ are close in Wasserstein-1 distance, then their spectral powers at $\omega$ are unlikely to differ much.}
The exponential term decays with the signal length $\ell$, while the second term scales linearly with $\mathcal W_1(\overline P_A,\overline P_B)$.

\subsubsection{From datasets to signals}
Let $T_A$ and $T_B$ be the data-sample distributions of $Q_A$ and $Q_B$ (in the original $d$-dimensional attribute space), respectively.
Recall that $\varphi$ bins points by their $u$-coordinates.
For any bin coordinate $c$, let $\Pr(u_A=c)$ (resp., $\Pr(u_B=c)$) be the probability that a sample from $T_A$ (resp., $T_B$) is mapped by $\phi_u$ to $u=c$.
Define $p=\min_{\forall c}\{\Pr(u_A=c),\Pr(u_B=c)\}$ and let $r$ be the maximum number of samples that fall into any bin in either dataset.

\begin{restatable}[Preservation of Wasserstein-1 Distance through Dataset-to-Signal Mapping]{theorem}{maintheoremB}
\label{thm:W_distance_relation}
If $\mathcal{W}_1(T_A, T_B) \leq \zeta$,
then
\[
\begin{aligned}
\Pr \left( \mathcal{W}_1(\overline{P}_A,\overline{P}_B)
\le\frac{C}{p}\left[ \frac{1}{s} \left( \frac{\zeta}{\beta} + 1 \right) + \zeta \right]
+2 \ell \sqrt{\frac{2V^2}{(1-\delta)rp}\log\frac{4 \ell}{\eta}} \right) \\
\geq
1 - 2 \ell \exp\Big(-\frac{\delta^2}{2}rp\Big) - \eta,
\end{aligned}
\]
where $C=\max\{1,4V\}$, $\beta>0$ and $s>0$ are the bin width and scaling factor in $\varphi$, and $\delta\in(0,1)$ and $\eta\in(0,1)$ are concentration parameters introduced by Chernoff's and Hoeffding's inequalities.
\end{restatable}

\Cref{thm:W_distance_relation} formalizes the \textbf{straight intuition}: \emph{small changes in the dataset distribution} (small $\mathcal W_1(T_A,T_B)$) \emph{typically lead to small changes in the signal distribution} (small $\mathcal W_1(\overline P_A,\overline P_B)$), up to finite-sample fluctuations controlled by $r$, $p$, and $\ell$.

\subsubsection{Putting the two theorems together}
Combining \Cref{thm:W_distance_relation} with \Cref{thm:ls-power-bound-bar} yields a \textbf{simple chain}:
\begin{align*}
\mathcal W_1(T_A,T_B)\ & \text{small}
\ \Longrightarrow\ 
\mathcal W_1(\overline P_A,\overline P_B)\ \text{small (with high probability)}\\
& \ \Longrightarrow\ 
|L_A(\omega)-L_B(\omega)|\ \text{small (with high probability)}.
\end{align*}
We now interpret this chain under attacks by setting $Q_A=Q_{wm}$ to be the watermarked dataset, and $Q_B=Q_{atk}$ to be the attacked dataset.

\begin{remark}[Utility–watermark trade-off]
\label{rem:dilemma}
An adversary aims to weanken watermark signal while preserving data utility. 
However, high data utility typically requires the attacked dataset $Q_B$ to remain distributionally close to the watermarked dataset $Q_A$, that is, $\mathcal W_1(T_A,T_B)$ must be small and the effective sample support across bins must remain sufficient. 
Under these conditions, \Cref{thm:W_distance_relation} implies that the signal distributions $\overline P_A$ and $\overline P_B$ remain close with high probability, and \Cref{thm:ls-power-bound-bar} further implies that the spectral powers $L_A(\omega)$ and $L_B(\omega)$ remain close with high probability. 
Since $Q_A$ has a large $L_A(\omega)$, $L_B(\omega)$ also remains large, which means the watermark stays detectable in $Q_B$.

Therefore, substantially weakening the watermark signal requires increasing $\mathcal W_1(T_A,T_B)$ or significantly reducing the sample support of bins, both of which distort the data distribution and typically degrade utility. 
This establishes an inherent \textbf{utility-watermark trade-off}: preserving utility tends to preserve the watermark, while suppressing the watermark requires sacrificing utility.
\end{remark}

\begin{remark}[Radioactivity against retraining attacks]
\label{rem:radioactivity}
A retraining attack trains a new generative model on $Q_A$ and outputs $Q_B$ by sampling from that model.
If retraining succeeds at preserving utility, then the generated distribution $T_B$ must approximate $T_A$, meaning $\mathcal W_1(T_A,T_B)$ is small.
By \Cref{thm:W_distance_relation}, this implies $\overline P_B$ remains close to $\overline P_A$ with high probability (w.h.p.); by \Cref{thm:ls-power-bound-bar}, this further implies $L_B(\omega)\approx L_A(\omega)$ with high probability.
Because the watermark in $Q_A$ creates an elevated spectral power peak at $\omega$ (large $L_A(\omega)$), the retrained model reproduces this elevated spectral power in $Q_B$ as part of matching the data distribution.
Hence, retraining preserves the watermark rather than removing it, which leads to exactly the radioactivity of RaMark.
\end{remark}

\begin{remark}[Robustness against data modification attacks]
\label{rem:robustness}
A data modification attack constructs $Q_B$ by directly perturbing $Q_A$ (e.g., adding noise, deleting/inserting rows, or deleting columns).
Such operations can be viewed as perturbing the underlying sample distribution from $T_A$ to a nearby $T_B$.
If the adversary keeps the perturbation mild to preserve utility, then $\mathcal W_1(T_A,T_B)$ remains small and the effective bin coverage (captured by $rp$) remains sufficiently large.
The same implication chain applies: $\overline P_B$ stays close to $\overline P_A$ (w.h.p.), and thus $L_B(\omega)$ stays close to $L_A(\omega)$ (w.h.p.), so the watermark remains detectable.
Therefore, to effectively suppress $L_B(\omega)$, the adversary must apply modifications strong enough to substantially change $T_B$ or severely reduce effective sample support across bins, which typically causes a noticeable utility drop.
\end{remark}


\section{Experiments}
\label{sec:exp}

In this section, we empirically evaluate the robustness and utility of RaMark under the threat models defined in \Cref{sec:threat_model}. Our experiments focus on the central question of this paper: whether the watermark remains detectable when an adversary attempts to remove it while preserving data utility. We compare RaMark with seven state-of-the-art watermarking methods on two real-world tabular datasets, under both retraining attacks and data modification attacks. All experiments were conducted on a desktop with an Intel(R) Core(TM) i9-10900K CPU and 64 GB RAM.

The remainder of this section is organized as follows. \Cref{sec:exp_set} describes the datasets, baselines, and basic experiment setup. \Cref{sec:perform_eval} introduces the evaluation protocols of watermark robustness. \Cref{sec:exp_robust_retrain}, \Cref{sec:exp_robust_common_atks}, and \Cref{sec:exp_params_ans} address the following questions, respectively:
\begin{enumerate}
    \item[\textbf{Q1}:] Is RaMark radioactive against retraining attacks?
    \item[\textbf{Q2}:] How robust is RaMark against data modification attacks?
    \item[\textbf{Q3}:] \pend{How do key parameters in RaMark affect its effectiveness?}
\end{enumerate}

\subsection{Experimental Setup}
\label{sec:exp_set}

\subsubsection{Datasets.}
We evaluate all methods on two widely used public tabular datasets.
\emph{Higgs Small (HS)}~\cite{vanschoren2014openml} is derived from high-energy physics experiments and is used for binary classification.  
\emph{House-16H (HO)}~\footnote{\url{https://www.dcc.fc.up.pt/~ltorgo/Regression/census.html}} is a regression dataset for predicting house prices from 16 attributes, such as house area, number of rooms, and furnishing status.  

For each dataset, we randomly select $25\%$ of the columns to carry the watermark, that is, seven columns for HS and four columns for HO. Basic statistics of the datasets are summarized in \Cref{tab:dataset}.

\begin{table}[t]
  \caption{Statistics of datasets. \#Train, \#Valid, and \#Test are the number of data samples in the training, validation, and test splits, respectively, and \#Attri is the number of attributes.}
  \label{tab:dataset}
  \resizebox{\columnwidth}{!}{%
  \begin{tabular}{l|cccccl}
    \toprule
    \textbf{Name}        & \textbf{\#Train} & \textbf{\#Valid} & \textbf{\#Test} & \textbf{\#Attri} & \textbf{ML Task}      \\
    \midrule
    Higgs Small & 62,751  & 15,688       & 19,610 & 28         & Classification \\
    House 16H   & 14,581  & 3,646        & 4,557  & 16         & Regression   \\ 
  \bottomrule
\end{tabular}%
}
\end{table}

\subsubsection{Baseline methods.}
We compare RaMark against seven representative watermarking methods for tabular data: S2R2W~\cite{li2022secure}, WGTD~\cite{he2024watermarking}, PKF~\cite{che2025primary}, TabWak~\cite{zhu2025tabwak}, TabularMark~\cite{zheng2024tabularmark}, MUSE~\cite{fang2025muse}, and B2Mark~\cite{zheng2025b2mark}.

We use the publicly available implementations for all baselines. The code for TabularMark, PKF, TabWak, MUSE, and B2Mark was released by the original authors. The implementations of S2R2W and WGTD were provided by the authors of TabularMark and MUSE, respectively. In our verification, all implementations reproduce the performance trends reported in their corresponding papers.

Since S2R2W and TabularMark are designed for single-column watermarking, we extend them to multi-column watermarking following~\cite{shehab2007watermarking,sebe2005noise,che2025primary}: the same watermark is independently embedded into multiple columns, and final detection is decided by majority voting across watermarked columns. For S2R2W, which requires a primary key, we follow TabularMark~\cite{zheng2024tabularmark} and use the top-10 most significant bits of one attribute as the primary key.

\subsubsection{Aligning watermark strength.}
\label{sec:align_wm_strength}
A fair robustness comparison requires that different methods operate under comparable watermark strength. However, watermark strength is controlled differently across methods and is not directly comparable.

We therefore use \emph{Machine Learning Efficiency} ($\mathtt{MLE}$)~\cite{kotelnikov2023tabddpm,zhu2025tabwak,al2015efficient,shi2024tabdiff} as a common, method-independent proxy. $\mathtt{MLE}$ measures how well generated data preserve statistical properties relevant for downstream tasks, under the standard \emph{Training on Generated and Testing on Real} protocol~\cite{fekri2019generating,kotelnikov2023tabddpm,kumari2024comparative}. If generated data accurately approximate the real distribution, models trained on them generalize well to real data, yielding high $\mathtt{MLE}$. Stronger watermark perturbations typically distort the data distribution more, resulting in lower $\mathtt{MLE}$~\cite{che2025primary,zhu2025tabwak}.

Based on the above observation, we align watermark strength across different watermarking methods by imposing an $\mathtt{MLE}$ budget constraint~\cite{che2025primary}:
\begin{equation}
    \mathtt{MLE}(Q_{uw}) - \mathtt{MLE}(Q_{wm}) \leq \gamma,
\end{equation}
where $Q_{uw}$ and $Q_{wm}$ denote the unwatermarked and watermarked generated datasets, respectively. The budget $\gamma$ controls the maximum allowable utility degradation caused by watermark embedding.
In our setting, $\gamma$ also serves as a control on perturbations introduced by the watermark, since a watermark is only accepted when the resulting utility loss remains small. Therefore, the embedded dependency is not an unconstrained distortion, but a utility-aware and controlled perturbation of the generated distribution.

Following~\citet{kotelnikov2023tabddpm}, we evaluate $\mathtt{MLE}$ using CatBoost~\cite{prokhorenkova2018catboost} on both the datasets of HS and HO. For the classification task on HS, we report F1-score as $\mathtt{MLE}$; for the regression task on HO, we report the coefficient of determination $R^2$ as $\mathtt{MLE}$~\cite{kotelnikov2023tabddpm,casella2024statistical,garthwaite2002statistical}.

In all experiments, we set $\gamma=1\%$. For each method, watermark-specific hyperparameters are tuned to achieve the strongest possible watermark under this $\mathtt{MLE}$ constraint, while other parameters remain at default settings. Because all methods are restricted to the same utility degradation, their watermark strengths are aligned, which enables a fair robustness comparison.

Following the literature~\cite{kotelnikov2023tabddpm,zheng2024tabularmark}, we use $\mathtt{MLE}$ to evaluate the data utility because we focus on ML-oriented data usage scenarios, where predictive performance is the primary concern.

\subsubsection{Measuring attack strength.}
Given a watermarked dataset $Q_{wm}$ and an attacked dataset $Q_{atk}$, we measure attack strength by the resulting utility degradation:
\begin{equation}
\mathtt{MLE}(Q_{wm}) - \mathtt{MLE}(Q_{atk}).
\end{equation}
This mirrors the watermark strength alignment principle: stronger attacks induce larger utility loss. Reporting attack strength alongside watermark robustness in our experimental results allows us to directly examine the utility–watermark trade-off predicted by our theoretical analysis in \Cref{sec:theory}.

\begin{table*}[t]
\caption{Radioactivity against retraining attacks on the datasets HS and HO.}
\label{tab:retrain}
\resizebox{\textwidth}{!}{%
\begin{tabular}{c|c|c|cc|cc|cc|cc}

\thicktoprule
\multirow{2}{*}{Dataset / Measure} &
\multirow{2}{*}{Watermark Method} &
  \multirow{2}{*}{$\mathtt{MLE}(Q_{wm})$} &
  \multicolumn{2}{c|}{TabDDPM} &
  \multicolumn{2}{c|}{TabDiff} &
  \multicolumn{2}{c|}{CTAB-GAN+} &
  \multicolumn{2}{c}{TVAE} \\ \cline{4-11} 
&     &       & \multicolumn{1}{c|}{$\mathtt{MLE}(Q_{atk})$}   & $\mathtt{Det\_AUC}$  & \multicolumn{1}{c|}{$\mathtt{MLE}(Q_{atk})$}   & $\mathtt{Det\_AUC}$  & \multicolumn{1}{c|}{$\mathtt{MLE}(Q_{atk})$}   & $\mathtt{Det\_AUC}$  & \multicolumn{1}{c|}{$\mathtt{MLE}(Q_{atk})$}   & $\mathtt{Det\_AUC}$  \\ \cline{1-11}
     
\multirow{9}{*}{HS / Detectability ($\mathtt{Det\_AUC}$)} &S2R2W & 0.703$\pm$0.009 & \multicolumn{1}{c|}{0.688$\pm$0.008} & 0.51 & \multicolumn{1}{c|}{0.680$\pm$0.003} & 0.50 & \multicolumn{1}{c|}{0.634$\pm$0.008} & 0.51 & \multicolumn{1}{c|}{0.608$\pm$0.006} & 0.51 \\ \cline{2-11}


&TabularMark & 0.703$\pm$0.008 & \multicolumn{1}{c|}{0.689$\pm$0.005} & 0.51 & \multicolumn{1}{c|}{0.679$\pm$0.010} & 0.51 & \multicolumn{1}{c|}{0.639$\pm$0.003} & 0.50 & \multicolumn{1}{c|}{0.607$\pm$0.005} & 0.51 \\ \cline{2-11}

&WGTD & 0.703$\pm$0.005 & \multicolumn{1}{c|}{0.684$\pm$0.011} & 0.51 & \multicolumn{1}{c|}{0.682$\pm$0.003} & 0.53 & \multicolumn{1}{c|}{0.630$\pm$0.011} & 0.57 & \multicolumn{1}{c|}{0.604$\pm$0.009} & 0.55 \\ \cline{2-11}

&PKF & 0.703$\pm$0.006 & \multicolumn{1}{c|}{0.685$\pm$0.004} & 0.52 & \multicolumn{1}{c|}{0.679$\pm$0.005} & 0.52 & \multicolumn{1}{c|}{0.635$\pm$0.009} & 0.51 & \multicolumn{1}{c|}{0.605$\pm$0.009} & 0.51 \\ \cline{2-11}

&TabWak & 0.704$\pm$0.005 & \multicolumn{1}{c|}{0.679$\pm$0.003} & 0.51 & \multicolumn{1}{c|}{0.678$\pm$0.003} & 0.55 & \multicolumn{1}{c|}{0.630$\pm$0.008} & 0.53 & \multicolumn{1}{c|}{0.599$\pm$0.008} & 0.53 \\ \cline{2-11}

&MUSE & 0.703$\pm$0.004 & \multicolumn{1}{c|}{0.680$\pm$0.003} & 0.50 & \multicolumn{1}{c|}{0.681$\pm$0.003} & 0.54 & \multicolumn{1}{c|}{0.629$\pm$0.009} & 0.54 & \multicolumn{1}{c|}{0.599$\pm$0.009} & 0.53 \\ \cline{2-11}

&B2Mark & 0.703$\pm$0.006 & \multicolumn{1}{c|}{0.679$\pm$0.006} & 0.60 & \multicolumn{1}{c|}{0.680$\pm$0.005} & 0.60 & \multicolumn{1}{c|}{0.630$\pm$0.005} & 0.56 & \multicolumn{1}{c|}{0.600$\pm$0.007} & 0.57 \\ \cline{2-11}

&\textbf{RaMark} & 0.704$\pm$0.004 & \multicolumn{1}{c|}{0.685$\pm$0.003} & \textbf{1.00} & \multicolumn{1}{c|}{0.680$\pm$0.006} & \textbf{0.98} & \multicolumn{1}{c|}{0.639$\pm$0.010} & \textbf{0.82} & \multicolumn{1}{c|}{0.605$\pm$0.009} & \textbf{0.79} \\ \midrule

\multirow{2}{*}{Dataset / Measure} &
\multirow{2}{*}{Watermark Method} &
  \multirow{2}{*}{$\mathtt{MLE}(Q_{wm})$} &
  \multicolumn{2}{c|}{TabDDPM} &
  \multicolumn{2}{c|}{TabDiff} &
  \multicolumn{2}{c|}{CTAB-GAN+} &
  \multicolumn{2}{c}{TVAE} \\ \cline{4-11} 
&     &       & \multicolumn{1}{c|}{$\mathtt{MLE}(Q_{atk})$}   & $\mathtt{Tra\_AUC}$  & \multicolumn{1}{c|}{$\mathtt{MLE}(Q_{atk})$}   & $\mathtt{Tra\_AUC}$  & \multicolumn{1}{c|}{$\mathtt{MLE}(Q_{atk})$}   & $\mathtt{Tra\_AUC}$  & \multicolumn{1}{c|}{$\mathtt{MLE}(Q_{atk})$}   & $\mathtt{Tra\_AUC}$  \\ \cline{1-11}
     
\multirow{9}{*}{HS / Traceability ($\mathtt{Tra\_AUC}$)} &S2R2W & 0.704$\pm$0.003 & \multicolumn{1}{c|}{0.687$\pm$0.008} & 0.03 & \multicolumn{1}{c|}{0.679$\pm$0.003} & 0.01 & \multicolumn{1}{c|}{0.635$\pm$0.008} & 0.02 & \multicolumn{1}{c|}{0.607$\pm$0.008} & 0.01 \\ \cline{2-11}


&TabularMark & 0.703$\pm$0.007 & \multicolumn{1}{c|}{0.689$\pm$0.005} & 0.50 & \multicolumn{1}{c|}{0.680$\pm$0.003} & 0.51 & \multicolumn{1}{c|}{0.633$\pm$0.003} & 0.51 & \multicolumn{1}{c|}{0.599$\pm$0.008} & 0.50 \\ \cline{2-11}

&WGTD & 0.703$\pm$0.008 & \multicolumn{1}{c|}{0.688$\pm$0.012} & 0.53 & \multicolumn{1}{c|}{0.679$\pm$0.003} & 0.51 & \multicolumn{1}{c|}{0.634$\pm$0.011} & 0.50 & \multicolumn{1}{c|}{0.606$\pm$0.008} & 0.50 \\ \cline{2-11}

&PKF & 0.703$\pm$0.012 & \multicolumn{1}{c|}{0.689$\pm$0.004} & 0.52 & \multicolumn{1}{c|}{0.681$\pm$0.005} & 0.53 & \multicolumn{1}{c|}{0.633$\pm$0.010} & 0.50 & \multicolumn{1}{c|}{0.598$\pm$0.009} & 0.54 \\ \cline{2-11}

&TabWak & 0.703$\pm$0.003 & \multicolumn{1}{c|}{0.679$\pm$0.003} & 0.51 & \multicolumn{1}{c|}{0.680$\pm$0.006} & 0.51 & \multicolumn{1}{c|}{0.629$\pm$0.008} & 0.54 & \multicolumn{1}{c|}{0.601$\pm$0.011} & 0.52 \\ \cline{2-11}

&MUSE & 0.703$\pm$0.003 & \multicolumn{1}{c|}{0.679$\pm$0.003} & 0.51 & \multicolumn{1}{c|}{0.681$\pm$0.005} & 0.50 & \multicolumn{1}{c|}{0.632$\pm$0.009} & 0.53 & \multicolumn{1}{c|}{0.601$\pm$0.005} & 0.53 \\ \cline{2-11}

&B2Mark & 0.703$\pm$0.005 & \multicolumn{1}{c|}{0.680$\pm$0.006} & 0.05 & \multicolumn{1}{c|}{0.680$\pm$0.002} & 0.05 & \multicolumn{1}{c|}{0.633$\pm$0.004} & 0.03 & \multicolumn{1}{c|}{0.599$\pm$0.005} & 0.04 \\ \cline{2-11}

&\textbf{RaMark} & 0.704$\pm$0.008 & \multicolumn{1}{c|}{0.689$\pm$0.009} & \textbf{1.00} & \multicolumn{1}{c|}{0.678$\pm$0.008} & \textbf{0.97} & \multicolumn{1}{c|}{0.634$\pm$0.010} & \textbf{0.79} & \multicolumn{1}{c|}{0.605$\pm$0.009} & \textbf{0.78} \\ \thickmidrule

\multirow{2}{*}{Dataset / Measure} &
\multirow{2}{*}{Watermark Method} &
  \multirow{2}{*}{$\mathtt{MLE}(Q_{wm})$} &
  \multicolumn{2}{c|}{TabDDPM} &
  \multicolumn{2}{c|}{TabDiff} &
  \multicolumn{2}{c|}{CTAB-GAN+} &
  \multicolumn{2}{c}{TVAE} \\ \cline{4-11} 
&     &       & \multicolumn{1}{c|}{$\mathtt{MLE}(Q_{atk})$}   & $\mathtt{Det\_AUC}$  & \multicolumn{1}{c|}{$\mathtt{MLE}(Q_{atk})$}   & $\mathtt{Det\_AUC}$  & \multicolumn{1}{c|}{$\mathtt{MLE}(Q_{atk})$}   & $\mathtt{Det\_AUC}$  & \multicolumn{1}{c|}{$\mathtt{MLE}(Q_{atk})$}   & $\mathtt{Det\_AUC}$  \\ \cline{1-11}

\multirow{9}{*}{HO / Detectability ($\mathtt{Det\_AUC}$)} & S2R2W & 0.610$\pm$0.013 & \multicolumn{1}{c|}{0.597$\pm$0.002} & 0.51 & \multicolumn{1}{c|}{0.591$\pm$0.006} & 0.51 & \multicolumn{1}{c|}{0.483$\pm$0.007} & 0.52 & \multicolumn{1}{c|}{0.471$\pm$0.008} & 0.51 \\ \cline{2-11}


&TabularMark & 0.610$\pm$0.010 & \multicolumn{1}{c|}{0.597$\pm$0.005} & 0.50 & \multicolumn{1}{c|}{0.589$\pm$0.005} & 0.51 & \multicolumn{1}{c|}{0.485$\pm$0.009} & 0.50 & \multicolumn{1}{c|}{0.470$\pm$0.009} & 0.51 \\ \cline{2-11}

&WGTD & 0.610$\pm$0.006 & \multicolumn{1}{c|}{0.598$\pm$0.001} & 0.50 & \multicolumn{1}{c|}{0.591$\pm$0.003} & 0.53 & \multicolumn{1}{c|}{0.482$\pm$0.006} & 0.50 & \multicolumn{1}{c|}{0.471$\pm$0.009} & 0.50 \\ \cline{2-11}

&PKF & 0.609$\pm$0.008 & \multicolumn{1}{c|}{0.598$\pm$0.004} & 0.50 & \multicolumn{1}{c|}{0.588$\pm$0.005} & 0.52 & \multicolumn{1}{c|}{0.481$\pm$0.008} & 0.51 & \multicolumn{1}{c|}{0.469$\pm$0.009} & 0.51 \\ \cline{2-11}

&TabWak & 0.610$\pm$0.005 & \multicolumn{1}{c|}{0.591$\pm$0.003} & 0.51 & \multicolumn{1}{c|}{0.588$\pm$0.003} & 0.55 & \multicolumn{1}{c|}{0.481$\pm$0.018} & 0.53 & \multicolumn{1}{c|}{0.471$\pm$0.012} & 0.53 \\ \cline{2-11}

&MUSE & 0.609$\pm$0.004 & \multicolumn{1}{c|}{0.592$\pm$0.003} & 0.50 & \multicolumn{1}{c|}{0.588$\pm$0.003} & 0.54 & \multicolumn{1}{c|}{0.481$\pm$0.009} & 0.54 & \multicolumn{1}{c|}{0.470$\pm$0.009} & 0.53 \\ \cline{2-11}

&B2Mark & 0.609$\pm$0.005 & \multicolumn{1}{c|}{0.596$\pm$0.007} & 0.59 & \multicolumn{1}{c|}{0.590$\pm$0.005} & 0.59 & \multicolumn{1}{c|}{0.480$\pm$0.019} & 0.56 & \multicolumn{1}{c|}{0.469$\pm$0.006} & 0.56 \\ \cline{2-11}

&\textbf{RaMark} & 0.610$\pm$0.005 & \multicolumn{1}{c|}{0.598$\pm$0.005} & \textbf{1.00} & \multicolumn{1}{c|}{0.590$\pm$0.012} & \textbf{0.98} & \multicolumn{1}{c|}{0.481$\pm$0.011} & \textbf{0.80} & \multicolumn{1}{c|}{0.470$\pm$0.010} & \textbf{0.79} \\ \midrule

\multirow{2}{*}{Dataset / Measure} &
\multirow{2}{*}{Watermark Method} &
  \multirow{2}{*}{$\mathtt{MLE}(Q_{wm})$} &
  \multicolumn{2}{c|}{TabDDPM} &
  \multicolumn{2}{c|}{TabDiff} &
  \multicolumn{2}{c|}{CTAB-GAN+} &
  \multicolumn{2}{c}{TVAE} \\ \cline{4-11} 
&     &       & \multicolumn{1}{c|}{$\mathtt{MLE}(Q_{atk})$}   & $\mathtt{Tra\_AUC}$  & \multicolumn{1}{c|}{$\mathtt{MLE}(Q_{atk})$}   & $\mathtt{Tra\_AUC}$  & \multicolumn{1}{c|}{$\mathtt{MLE}(Q_{atk})$}   & $\mathtt{Tra\_AUC}$  & \multicolumn{1}{c|}{$\mathtt{MLE}(Q_{atk})$}   & $\mathtt{Tra\_AUC}$  \\ \cline{1-11}
     
\multirow{9}{*}{HO / Traceability ($\mathtt{Tra\_AUC}$)} & S2R2W & 0.609$\pm$0.002 & \multicolumn{1}{c|}{0.597$\pm$0.004} & 0.02 & \multicolumn{1}{c|}{0.588$\pm$0.003} & 0.01 & \multicolumn{1}{c|}{0.483$\pm$0.012} & 0.02 & \multicolumn{1}{c|}{0.469$\pm$0.008} & 0.01 \\ \cline{2-11}


&TabularMark & 0.610$\pm$0.007 & \multicolumn{1}{c|}{0.597$\pm$0.002} & 0.50 & \multicolumn{1}{c|}{0.591$\pm$0.003} & 0.51 & \multicolumn{1}{c|}{0.486$\pm$0.018} & 0.50 & \multicolumn{1}{c|}{0.473$\pm$0.008} & 0.50 \\ \cline{2-11}

&WGTD & 0.610$\pm$0.008 & \multicolumn{1}{c|}{0.598$\pm$0.011} & 0.50 & \multicolumn{1}{c|}{0.592$\pm$0.003} & 0.51 & \multicolumn{1}{c|}{0.481$\pm$0.007} & 0.50 & \multicolumn{1}{c|}{0.469$\pm$0.009} & 0.51 \\ \cline{2-11}

&PKF & 0.608$\pm$0.009 & \multicolumn{1}{c|}{0.598$\pm$0.005} & 0.50 & \multicolumn{1}{c|}{0.589$\pm$0.005} & 0.50 & \multicolumn{1}{c|}{0.481$\pm$0.008} & 0.50 & \multicolumn{1}{c|}{0.471$\pm$0.008} & 0.50 \\ \cline{2-11}

&TabWak & 0.610$\pm$0.005 & \multicolumn{1}{c|}{0.588$\pm$0.003} & 0.51 & \multicolumn{1}{c|}{0.592$\pm$0.002} & 0.51 & \multicolumn{1}{c|}{0.481$\pm$0.008} & 0.53 & \multicolumn{1}{c|}{0.471$\pm$0.009} & 0.50 \\ \cline{2-11}

&MUSE & 0.609$\pm$0.005 & \multicolumn{1}{c|}{0.589$\pm$0.003} & 0.51 & \multicolumn{1}{c|}{0.589$\pm$0.003} & 0.50 & \multicolumn{1}{c|}{0.480$\pm$0.009} & 0.52 & \multicolumn{1}{c|}{0.469$\pm$0.008} & 0.53 \\ \cline{2-11}

&B2Mark & 0.609$\pm$0.005 & \multicolumn{1}{c|}{0.596$\pm$0.007} & 0.06 & \multicolumn{1}{c|}{0.590$\pm$0.005} & 0.05 & \multicolumn{1}{c|}{0.480$\pm$0.019} & 0.04 & \multicolumn{1}{c|}{0.469$\pm$0.006} & 0.04 \\ \cline{2-11}

&\textbf{RaMark} & 0.610$\pm$0.007 & \multicolumn{1}{c|}{0.598$\pm$0.005} & \textbf{1.00} & \multicolumn{1}{c|}{0.589$\pm$0.007} & \textbf{0.97} & \multicolumn{1}{c|}{0.481$\pm$0.014} & \textbf{0.80} & \multicolumn{1}{c|}{0.469$\pm$0.006} & \textbf{0.79} \\ \thickbottomrule

\end{tabular}%
}
\end{table*}

\subsection{Evaluation Protocol of Robustness}
\label{sec:perform_eval}

We evaluate watermark robustness by simulating a watermark service provider (WSD) that manages $10^5$ data owners. Each owner registers a unique secret key with the WSD and embeds a watermark using that key. Given a suspicious dataset, the WSD performs verification in two stages: (i) a \emph{detecting stage}, which determines whether the dataset contains a watermark, and (ii) a \emph{tracing stage}, which identifies the owner among the $10^5$ registered data owners if a watermark is detected in the detection stage.


\subsubsection{Detectability.}
Detectability measures how well a method distinguishes watermarked datasets from unwatermarked ones. We model the detecting stage as a binary classification task over 100 datasets: 50 watermarked (positive case) and 50 unwatermarked (negative case). 

The unwatermarked datasets are independently generated by TabDDPM~\cite{kotelnikov2023tabddpm}. 
Each watermarked dataset is generated by using a distinct secret key sampled uniformly from the $10^5$ registered keys.
For post-hoc methods (S2R2W~\cite{li2022secure}, TabularMark~\cite{zheng2024tabularmark}, WGTD~\cite{he2024watermarking}, PKF~\cite{che2025primary}, B2Mark~\cite{zheng2025b2mark}), which embed watermark after data generation, we first generate datasets using TabDDPM and then embed watermarks. 
For generative methods (TabWak~\cite{zhu2025tabwak}, MUSE~\cite{fang2025muse}), which embed watermark during data generation, watermarked datasets are generated directly using their default procedures. 

To evaluate detectability,
we compute $10^5$ detection scores for each dataset against all $10^5$ keys and take the maximum score as the \emph{final score}, representing the strongest evidence of watermark presence. Detectability is quantified by the Area Under the ROC Curve ($\mathtt{AUC}$) computed from the 100 final scores of the 100 datasets, denoted as $\mathtt{Det\_AUC}$. A higher $\mathtt{Det\_AUC}$ indicates better separation between watermarked and unwatermarked datasets.

\subsubsection{Traceability.}
Traceability measures how well a method identifies the correct owner among all the data owners. We construct a binary classification task over 100 watermarked datasets including 50 positive cases and 50 negative cases. Each dataset is watermarked using a distinct secret key sampled from the $10^5$ registered keys.

For a positive case, the final score is the detection score computed using its ground-truth key. For a negative case, the final score is the maximum detection score obtained using all incorrect keys (i.e., the strongest false match among the remaining $10^5 - 1$ keys). Traceability is quantified by the $\mathtt{AUC}$ of the ROC curve computed from these 100 final scores, denoted as $\mathtt{Tra\_AUC}$. A higher $\mathtt{Tra\_AUC}$ indicates stronger owner discrimination.

\subsubsection{Evaluation with varying thresholds.}
\pend{
As described in \Cref{sec:wm_detect}, the final decision on the presence of a watermark is obtained by comparing $\mathtt{DS}(\omega)$ with a predefined threshold.
In our evaluation, we use $\mathtt{AUC}$ as the evaluation metric, which summarizes detection performance over all possible thresholds.
This is particularly useful in our setting, as different watermarking methods may produce detection scores with different scales.
By measuring performance across varying thresholds, $\mathtt{AUC}$ enables a more comprehensive and fair comparison of both detectability and traceability across methods.
Throughout the remainder of the paper, we use $\mathtt{AUC}$ to refer to either $\mathtt{Det\_AUC}$ or $\mathtt{Tra\_AUC}$ when the context is clear.
}

\subsubsection{Evaluating robustness under attacks}
To evaluate robustness under attacks, we apply each attack described in \Cref{sec:wmr_attacks} to the watermarked datasets in the positive cases before computing detection scores. 
This setting applies to both the computation of $\mathtt{Det\_AUC}$ and $\mathtt{Tra\_AUC}$.

\subsection{Radioactivity against Retraining Attack}
\label{sec:exp_robust_retrain}

We evaluate robustness under retraining attack using the four generative models introduced in \Cref{sec:wmr_attacks}: TabDDPM, TabDiff, CTAB-GAN+, and TVAE. 
A higher robustness against retraining attack means a higher radioactivity of the watermark.
For each watermarking method, we retrain the attack model on the watermarked dataset $Q_{wm}$ and generate a new dataset $Q_{atk}$. The results on HS and HO are summarized in \Cref{tab:retrain}.

As it is shown in \Cref{tab:retrain}, $\mathtt{MLE}(Q_{wm})$ reports the mean $\pm$ standard deviation over the 50 watermarked datasets in the positive cases, and $\mathtt{MLE}(Q_{atk})$ reports the corresponding results after performing retraining attack. Because all methods are constrained by the same $\mathtt{MLE}$ budget during embedding, the $\mathtt{MLE}(Q_{wm})$ values are closely aligned across methods. This confirms that watermark strengths are comparable. Similarly, for each attack model, the $\mathtt{MLE}(Q_{atk})$ values across methods are also close, since the same retraining procedure is applied. Therefore, the comparison in \Cref{tab:retrain} reflects robustness differences rather than discrepancies in watermark strength or attack strength.

\subsubsection{Main observation.}
RaMark consistently achieves the highest $\mathtt{AUC}$ under all four retraining models and on both datasets. In particular, RaMark maintains $\mathtt{Det\_AUC}$ close to $1.00$ for TabDDPM and TabDiff, and remains substantially higher than all baselines even under CTAB-GAN+ and TVAE. This empirical result directly validates the theoretical analysis in \Cref{rem:dilemma,rem:radioactivity}: when retraining preserves data utility (small distributional deviation), the elevated spectral power at frequency $\omega$ is preserved, and the detection score remains high.

\subsubsection{Why baselines fail?}
The failure of baseline methods can be traced back to how their watermark signals are embedded.

S2R2W and TabularMark rely on (virtual) primary keys derived from specific attribute values. After retraining, newly generated samples differ from the original ones, so the corresponding primary keys cannot be reconstructed. As a result, watermark localization fails and detection performance collapses.

PKF embeds weak artificial correlations into the data distribution. Although these correlations are statistically detectable in the original watermarked dataset, they form only a minor component of the distribution and are not reliably reproduced by retrained generative models. Consequently, the spectral signal weakens after retraining.

WGTD embeds the watermark through independent numeric perturbations on selected values. Retrained models do not reproduce these exact perturbations, and thus the watermark signal disappears.

TabWak depends on a deterministic mapping between latent seeds and generated samples. Retraining breaks this mapping, so the watermark cannot be reconstructed in $Q_{atk}$.

MUSE embeds the watermark through a generation-time selection rule based on discrete bucket preferences. While this creates bias in the generated dataset, the rule itself is not encoded as a stable global dependency. Retrained models approximate the overall distribution but do not preserve this procedural selection mechanism, causing the watermark to vanish.

\subsubsection{Detectability vs.\ traceability.}
In \Cref{tab:retrain}, S2R2W and B2Mark exhibit $\mathtt{Tra\_AUC}$ values close to zero under retraining. Both methods encode the watermark as a bit sequence and compute detection scores using normalized Hamming similarity. After retraining, the recovered bit sequence is severely distorted, so the score under the correct key (positive case) becomes lower than the maximum score among incorrect keys (negative cases). This reversal drives $\mathtt{Tra\_AUC}$ toward zero. 

By contrast, their $\mathtt{Det\_AUC}$ values remain near $0.5$, because in detectability evaluation both positive and negative cases use the maximum score over all keys, resulting in statistically indistinguishable score distributions after retraining.

\subsubsection{Effect of retraining strength.}
Across TabDDPM, TabDiff, CTAB-GAN+, and TVAE, the $\mathtt{MLE}(Q_{atk})$ values decrease progressively, indicating increasing distributional distortion caused by retraining. As predicted by the theoretical trade-off in \Cref{sec:theory}, the $\mathtt{AUC}$ values of RaMark also decrease as $\mathtt{MLE}(Q_{atk})$ decreases. However, even when $\mathtt{MLE}(Q_{atk})$ drops by approximately $10\%\sim 15\%$ relative to $\mathtt{MLE}(Q_{wm})$, RaMark still maintains substantially higher $\mathtt{AUC}$ than all baselines. This confirms that the watermark embedded by RaMark is tightly coupled to the data distribution and therefore persists under distribution-preserving retraining.

\begin{figure}[t]
\centering
     \includegraphics[width=1.0\columnwidth]{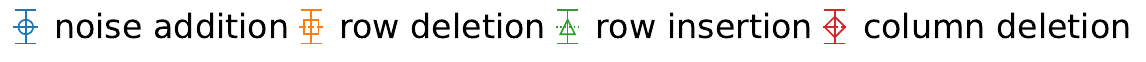}\\
    \vspace{\figurevreduce}
    \subfigure[$\mathtt{MLE}$ on HS]{
        \includegraphics[width=\figurewidthfour]{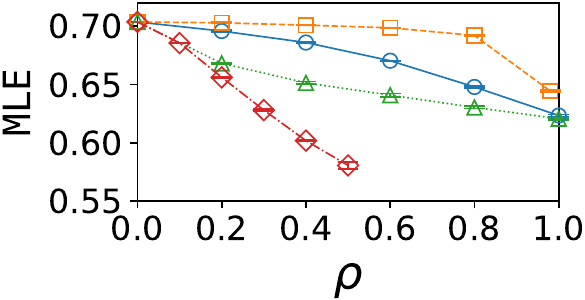}
	\label{fig:detect_MLE_uniform_alteration_higgs}
    }
    \subfigure[$\mathtt{MLE}$ on HO]{
        \includegraphics[width=\figurewidthfour]{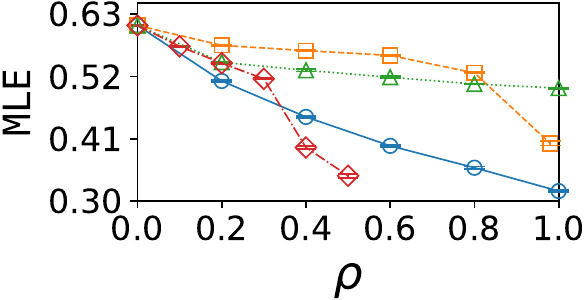}
	\label{fig:detect_MLE_uniform_alteration_house}
    }
\caption{The mean and standard deviation of the $\mathtt{MLE}(Q_{wm})$ and $\mathtt{MLE}(Q_{atk})$ for all the watermarking methods in the detectability experiments. 
The $x$-axis shows the attack strength $\rho\in\{\rho_{na}, \rho_{rd}, \rho_{ri}, \rho_{cd}\}$ of each data modification attack.
The $y$-axis shows the $\mathtt{MLE}$ for each attack. 
The leftmost point on each curve shows $\mathtt{MLE}(Q_{wm})$ when the attack strength is zero, which means watermarked datasets are not attacked. 
The other points show $\mathtt{MLE}(Q_{atk})$.
The legend shows the markers for different data modification attacks.
The standard deviation is shown by the error bar of each point on the curves, which can be better viewed when zoomed in.
}
\Description{}
\label{fig:MLE_manipulation_atk}
\end{figure}

\begin{figure*}[t]
\centering
     \includegraphics[height=0.045\columnwidth]{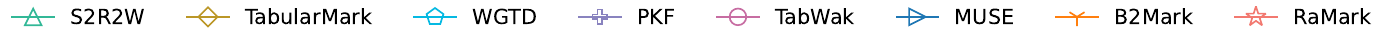}\\
     \vspace{\figurevreduce}
    \subfigure[Noise Addition, HS]{
        \includegraphics[width=\figurewidthone]{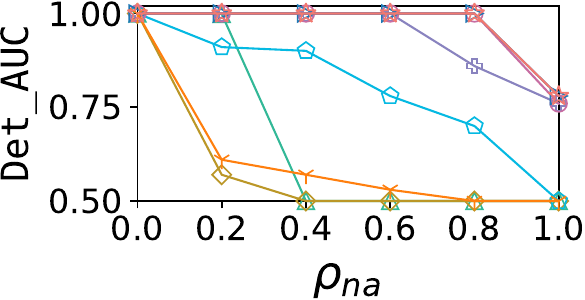}
	\label{fig:detect_auc_uniform_alteration_higgs}
    }
    \hspace{\figurehreduce}
    \subfigure[Noise Addition, HO]{
        \includegraphics[width=\figurewidthone]{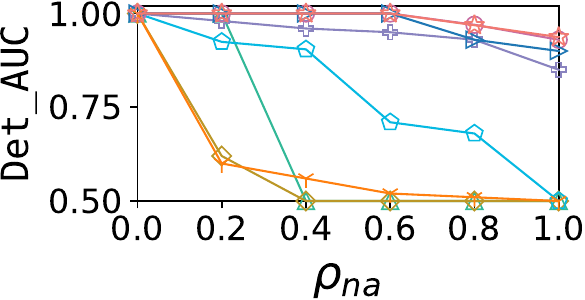}
	\label{fig:detect_auc_uniform_alteration_house}
    }
    \subfigure[Row Deletion, HS]{
        \includegraphics[width=\figurewidthone]{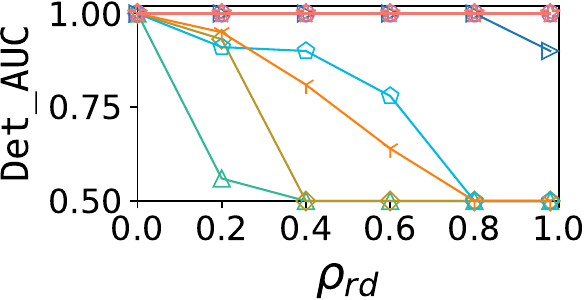}
        \label{fig:detect_auc_row_deletion_higgs}
    }
    \hspace{\figurehreduce}
    \subfigure[Row Deletion, HO]{
        \includegraphics[width=\figurewidthone]{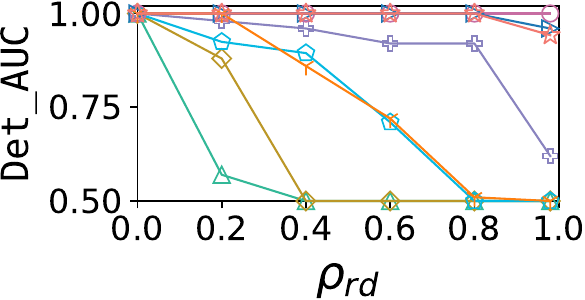}
        \label{fig:detect_auc_row_deletion_house}
    }
    \subfigure[Row Insertion, HS]{
        \includegraphics[width=\figurewidthone]{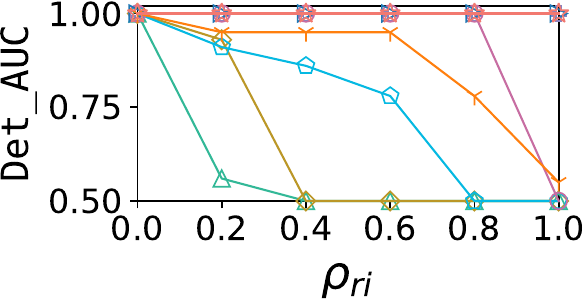}
        \label{fig:detect_auc_row_insertion_higgs}
    }
    \hspace{\figurehreduce}
    \subfigure[Row Insertion, HO]{
        \includegraphics[width=\figurewidthone]{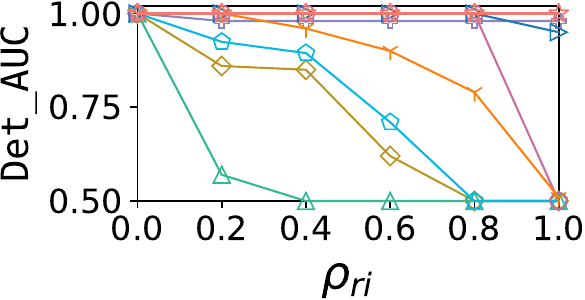}
        \label{fig:detect_auc_row_insertion_house}
    }
    \subfigure[Column Deletion, HS]{
        \includegraphics[width=\figurewidthone]{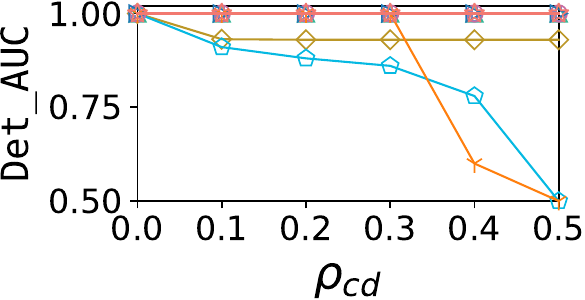}
        \label{fig:detect_auc_column_deletion_higgs}
    }
    \hspace{\figurehreduce}
    \subfigure[Column Deletion, HO]{
        \includegraphics[width=\figurewidthone]{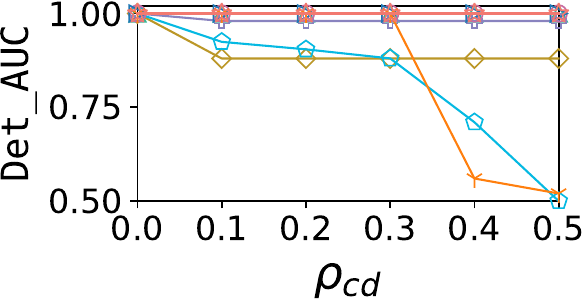}
        \label{fig:detect_auc_column_deletion_house}
    }
\caption{The $\mathtt{Det\_{AUC}}$ under different attack strengths. The attack strengths $\rho_{na}, \rho_{rd}, \rho_{ri}, \rho_{cd}$ are defined in 
\Cref{sec:wmr_attacks}.}
\Description{}
\label{fig:detect_auc}
\end{figure*}

\begin{figure*}[t]
\centering
     \includegraphics[height=0.045\columnwidth]{figures/man_atk_results/legend_auc}\\
     \vspace{\figurevreduce}
    \subfigure[Noise Addition, HS]{
        \includegraphics[width=\figurewidthone]{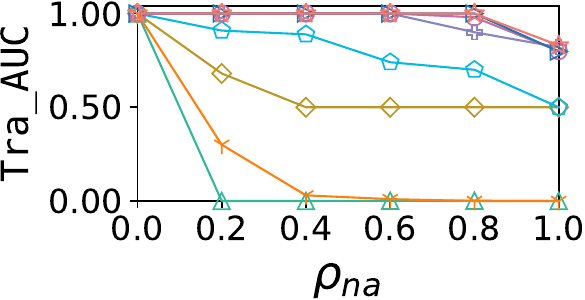}
	\label{fig:trace_auc_uniform_alteration_higgs}
    }
    \hspace{\figurehreduce}
    \subfigure[Noise Addition, HO]{
        \includegraphics[width=\figurewidthone]{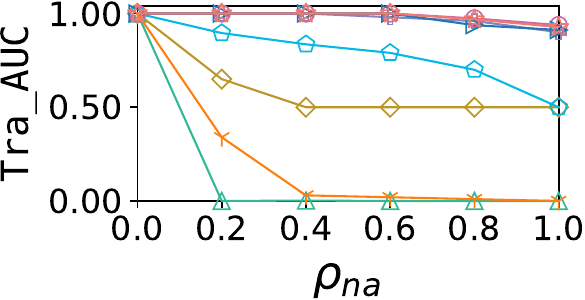}
	\label{fig:trace_auc_uniform_alteration_house}
    }
    \subfigure[Row Deletion, HS]{
        \includegraphics[width=\figurewidthone]{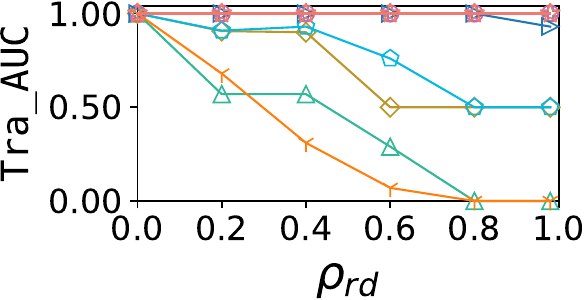}
        \label{fig:trace_auc_row_deletion_higgs}
    }
    \hspace{\figurehreduce}
    \subfigure[Row Deletion, HO]{
        \includegraphics[width=\figurewidthone]{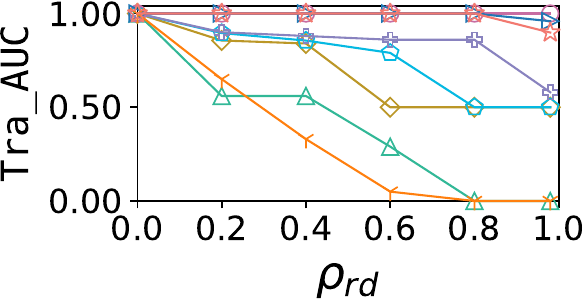}
        \label{fig:trace_auc_row_deletion_house}
    }
    \subfigure[Row Insertion, HS]{
        \includegraphics[width=\figurewidthone]{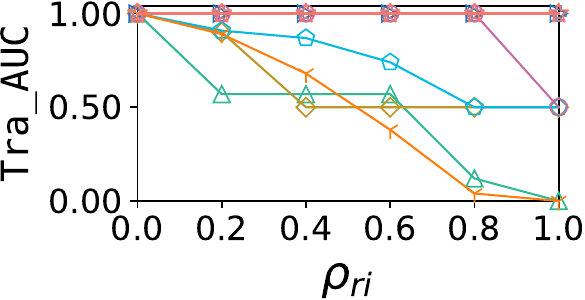}
        \label{fig:trace_auc_row_insertion_higgs}
    }
    \hspace{\figurehreduce}
    \subfigure[Row Insertion, HO]{
        \includegraphics[width=\figurewidthone]{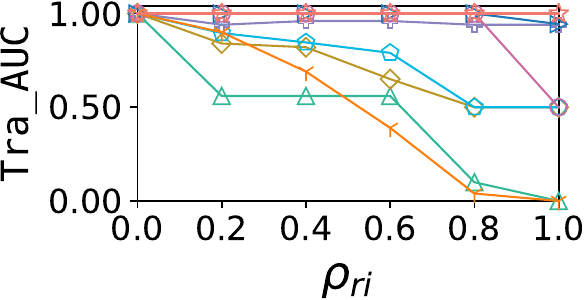}
        \label{fig:trace_auc_row_insertion_house}
    }
    \subfigure[Column Deletion, HS]{
        \includegraphics[width=\figurewidthone]{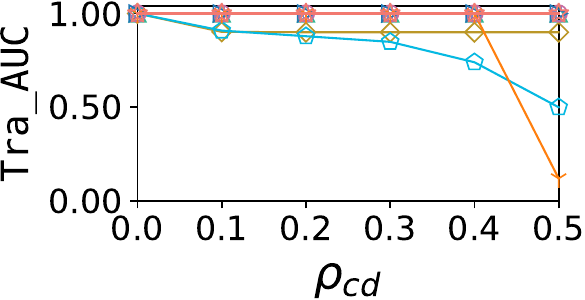}
        \label{fig:trace_auc_column_deletion_higgs}
    }  
    \hspace{\figurehreduce}
    \subfigure[Column Deletion, HO]{
        \includegraphics[width=\figurewidthone]{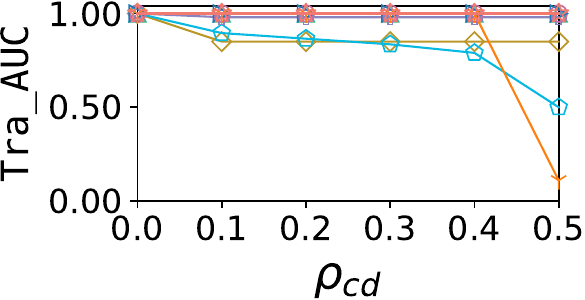}
        \label{fig:trace_auc_column_deletion_house}
    }
\caption{The $\mathtt{Tra\_{AUC}}$ under different attack strengths. The attack strengths $\rho_{na}, \rho_{rd}, \rho_{ri}, \rho_{cd}$ are defined in 
\Cref{sec:wmr_attacks}.}
\Description{}
\label{fig:trace_auc}
\end{figure*}

\subsection{Robustness to Data Modification Attacks}
\label{sec:exp_robust_common_atks}

We next evaluate robustness under the four data modification attacks defined in \Cref{sec:wmr_attacks}, including noise addition (na), row deletion (rd), row insertion (ri), and column deletion (cd).


\subsubsection{Utility under attack.}
\Cref{fig:MLE_manipulation_atk} reports $\mathtt{MLE}$ performance when attack strength increases. 
The curves from detectability and traceability experiments are nearly identical because they are derived from watermarked datasets with the same underlying distribution. 
For clarity, we therefore report only the $\mathtt{MLE}$ curves from the detectability experiments. 
Moreover, since all watermarking methods exhibit comparable $\mathtt{MLE}$ under the same attack on the same dataset, we report the mean and standard deviation of $\mathtt{MLE}$ across all methods for each attack setting.


%
%

We can see from \Cref{fig:MLE_manipulation_atk} that the following three observations are consistent across all attacks. 
First, $\mathtt{MLE}(Q_{wm})$ values are nearly identical across different watermarking methods, which confirms that the watermark strengths are aligned by the same $\mathtt{MLE}$ budget constraint. 
Second, $\mathtt{MLE}(Q_{atk})$ values are comparable with small standard deviation across watermarking methods under the same attack configuration, which means comparable attack strength is applied to all methods and thus ensures fair comparison of their watermark robustness. 
Third, as attack strength increases, $\mathtt{MLE}$ consistently decreases for all methods, which indicates stronger distortion of the data distribution. In practice, this means that increasing attack strength inevitably reduces data utility, which reinforces the trade-off discussed in \Cref{rem:dilemma}.

\subsubsection{Detectability and traceability under attack.}
\Cref{fig:detect_auc} and \Cref{fig:trace_auc} report $\mathtt{Det\_AUC}$ and $\mathtt{Tra\_AUC}$, respectively, under varying attack strengths. We analyze each attack type below.

\paragraph{Noise addition.}
Noise addition perturbs every attribute value with random noise. Methods that rely on precise numeric perturbations of selected table entries (S2R2W, TabularMark, WGTD, B2Mark) exhibit rapidly decreasing $\mathtt{AUC}$, since random noise directly disrupts their pointwise modifications. In contrast, PKF, TabWak, MUSE, and RaMark maintain higher $\mathtt{AUC}$ because their watermark signals are distributed across global statistical structure rather than individual values. Since moderate noise affects local values more than global dependencies, these methods are more robust.

\paragraph{Row deletion.}
Row deletion removes a fraction of samples. S2R2W degrades quickly because it modifies only a subset of samples; deleting these samples directly removes watermark evidence. TabularMark, WGTD, and B2Mark also degrade because their detection relies on statistical aggregation over selected ranges; reducing sample size increases variance and weakens detection reliability. By contrast, PKF, TabWak, MUSE, and RaMark distribute watermark signals across many rows through attribute-level dependencies. Partial row removal preserves much of this structure, resulting in stronger robustness.

\paragraph{Row insertion.}
Row insertion adds new samples. S2R2W and TabularMark are vulnerable because they depend on virtual primary keys (VPKs); inserted rows may introduce key conflicts, leading to misidentification of watermarked samples~\cite{gort2020double}. WGTD and B2Mark also degrade because inserted rows alter attribute frequencies within watermark-encoded ranges. In contrast, PKF, TabWak, MUSE, and RaMark do not rely on VPKs and aggregate watermark signals across all rows. Inserted rows dilute but do not eliminate the embedded structural dependencies, so $\mathtt{AUC}$ remains comparatively high.

\paragraph{Column deletion.}
Column deletion removes a fraction of attributes. WGTD, TabularMark, and B2Mark degrade because their detection relies on selected attribute subsets; removing these columns reduces valid statistical evidence. S2R2W remains relatively stable because it embeds watermark redundantly across columns, allowing detection as long as at least one watermarked column survives. PKF, TabWak, MUSE, and RaMark also maintain high $\mathtt{AUC}$ because their watermark is encoded through joint correlations in latent representations. Removing part of the attributes weakens but does not eliminate these correlations, thus the watermark signal is still preserved.

\subsubsection{Summary.}
Across all four data modification attacks and both datasets, RaMark achieves the highest or near-highest $\mathtt{AUC}$ under comparable utility degradation. These empirical results align with the theoretical analysis in \Cref{rem:dilemma} and \Cref{rem:robustness}, that is, as long as the attacked dataset remains distributionally close to the watermarked dataset to preserve data utility, the watermark signal, which is embedded as a stable distribution-level sinusoidal dependency among attributes, remains detectable.

   

\begin{figure}[t]
\centering
    \includegraphics[height=0.066\columnwidth]{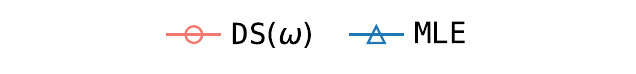}\\
     \vspace{\figurehreducetest}
    \subfigure[HS]{
       \includegraphics[width=\figurewidthfive]{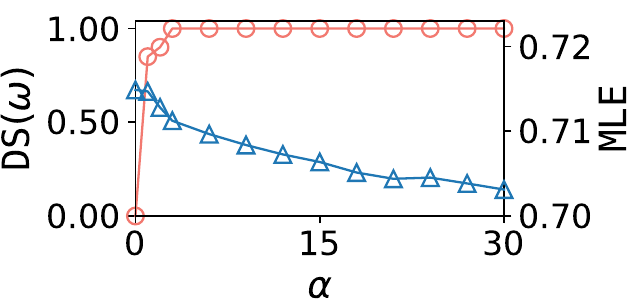}
       \label{fig:param_ans_alpha_hs}
   }
\hspace{\figurehreducenew}
   \subfigure[HO]{
       \includegraphics[width=\figurewidthfive]{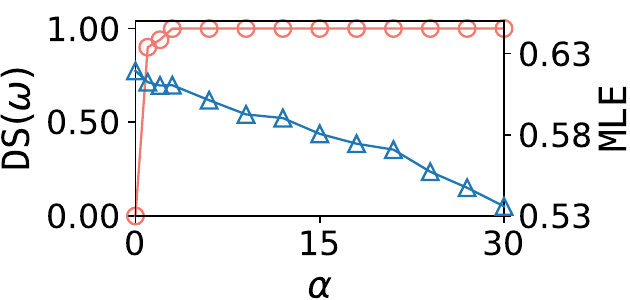}
       \label{fig:param_ans_alpha_ho}
   }
\caption{The effect of the guidance strength $\alpha$. }
\label{fig:param_ans_s}
\end{figure}

\subsection{\pend{Parameter Analysis}}
\label{sec:exp_params_ans}
\pend{
RaMark involves two types of parameters with distinct roles. 
The first is the \emph{guidance strength} $\alpha$ in \Cref{eq:final_zt}, which controls the strength of watermark guidance during the diffusion sampling process.
The second is a set of \emph{mapping parameters}, including the bin width $\beta$ and the scaling factor $s$ in \Cref{eq:mapping_x}, which determine how data samples are mapped into the discrete-time signal used for spectral analysis. 

In this subsection, we focus on analyzing $\alpha$. 
The sensitivity of $\beta$ and $s$ is provided in Appendix~\ref{app:param_analysis_beta_s}, where we show that RaMark remains stable across a wide range of these signal-construction parameters.
}

As shown in \Cref{fig:param_ans_s}, we study how $\alpha$ affects the trade-off between watermark detectability and data utility. we report both the detection score $DS(\omega)$ and $\mathtt{MLE}$ under different values of $\alpha$ on HS and HO.

When $\alpha$ increases, $\mathtt{DS}(\omega)$ increases monotonically and quickly approaches $1.0$. This behavior follows directly from \Cref{eq:final_zt}: a larger $\alpha$ amplifies the mean shift $\alpha(\Sigma \cdot g_z)$ in each reverse denoising step, which steers generated samples closer to the predefined sinusoidal dependency. Consequently, the spectral power $L(\omega)$ at frequency $\omega$ increases, the false alarm probability $\mathtt{FAP}(\omega)$ decreases, and the detection score $\mathtt{DS}(\omega)=1-\mathtt{FAP}(\omega)$ increases. When $\alpha$ is sufficiently large, the embedded watermark becomes highly prominent and is detected with near-perfect confidence.

At the same time, $\mathtt{MLE}$ gradually decreases as $\alpha$ increases. 
This is because increasing $\alpha$ strengthens the watermark guidance term and shifts samples further away from the distribution learned from $Q_{ori}$. While this improves watermark alignment, it also introduces distributional distortion that reduces sample fidelity and degrades the performance of downstream learning tasks, which manifests as lower $\mathtt{MLE}$.

Overall, the results demonstrate a clear and controlled trade-off governed by $\alpha$. Smaller $\alpha$ preserves data utility (high $\mathtt{MLE}$) but yields weaker watermark signals (lower $\mathtt{DS}(\omega)$), whereas larger $\alpha$ strengthens watermark signal at the cost of reduced data utility. This trade-off between watermark and data utility is consistent with prior watermarking studies~\cite{che2025primary,zheng2025b2mark,kamran2018comprehensive,kumar2020recent} and aligns with the theoretical utility–watermark trade-off established in \Cref{sec:theory}.
\pend{In practice, setting $\alpha$ in the range of $[5, 10]$ is sufficient to achieve strong detectability while maintaining good data utility on both datasets.}

\section{Conclusion}
\label{sec:conclusions}


In this paper, we study radioactive watermarking for generated tabular data under retraining attacks, where an adversary trains a new generative model on a watermarked dataset and regenerates data that removes the watermark signal while preserving data utility. 
This problem arises in machine learning systems for generated data sharing and trading, especially in privacy-sensitive domains where the original data cannot be released directly. 
In such settings, watermarking provides a principled mechanism for ownership verification through detection and traceability after sharing or redistribution. 
We show that existing watermarking methods fail to achieve radioactivity, as their signals are not embedded as intrinsic distributional properties and are therefore not preserved under retraining.

To address this challenge, we propose RaMark, which embeds watermark signals as sinusoidal dependencies in the data distribution through watermark-guided diffusion sampling.
Our theoretical analysis establishes a principled link between distribution preservation, spectral power preservation, and watermark robustness. 
In particular, we show that preserving data utility implies persistence of the watermark signal, which leads to an inherent utility–watermark trade-off for adversaries. 
Extensive experiments on two real-world datasets demonstrate that RaMark consistently achieves substantially higher detectability and traceability than state-of-the-art baselines while operating under the same data utility budget under both retraining attacks and data modification attacks.

\pend{RaMark currently focuses on tabular datasets with at least two continuous-valued attributes. Extending radioactive watermarking to categorical or textual attributes remains an important direction for future work.}
\bibliographystyle{ACM-Reference-Format}
\bibliography{references}

\clearpage
\appendix




\section{Appendix}
\label{sec:appendix}

\subsection{Theoretical Analysis and Proof}
\label{app:the_proof}

\NameA*

\begin{proof}
By the definition of $p_\theta(z_t \mid z_{t+1})$, we have
\begin{equation}
\log p_\theta(z_t \mid z_{t+1}) = -\frac{1}{2} (z_t - \mu)^T \Sigma^{-1} (z_t - \mu) + A, 
\end{equation}
where $A$ is a constant.


Then we approximate $\log \Pr(W \mid z_t)$ using a Taylor expansion around $z_t=\mu$ as
\begin{equation}
\begin{aligned}
    &\log \Pr(W \mid z_t) \\
    \approx & \log \Pr(W \mid z_t) \big|_{z_t=\mu} + (z_t - \mu)
      \nabla_{z_t} \log \Pr(W \mid z_t)\big|_{z_t=\mu} \\
    = & (z_t - \mu) g_z + A_1.
\end{aligned}
\end{equation}
Here, $g_z =\nabla_{z_t} \log \Pr(W \mid z_t)\big|_{z_t=\mu}$, and $A_1$ is a constant. This gives
\begin{equation}
\begin{aligned}
&\log \big( p_\theta(z_t \mid z_{t+1}) \Pr(W \mid z_t) \big) \\
\approx
&-\frac{1}{2}(z_t - \mu)^{T} \Sigma^{-1} (z_t - \mu)
    + (z_t - \mu) g_z + A_2 \\
=& -\frac{1}{2} (z_t - \mu - \Sigma \cdot  g_z)^{T}
        \Sigma^{-1}
        (z_t - \mu - \Sigma \cdot  g_z)
+ \frac{1}{2} g^{\top} \Sigma \cdot  g_z + A_2 \\
= & -\frac{1}{2} (z_t - \mu - \Sigma \cdot  g_z)^{T}
        \Sigma^{-1}
        (z_t - \mu - \Sigma \cdot  g_z)
+ A_3 \\
\end{aligned}
\end{equation}
The quadratic term in the above expression is exactly the exponent of a Gaussian density with mean $\mu+\Sigma \cdot  g_z$ and covariance $\Sigma$. 
We have thus found that the conditional reverse transition can be approximated by a Gaussian similar to the unconditional reverse transition, but with its mean shifted by $\Sigma \cdot  g_z$.
\end{proof}

\maintheoremA*

\begin{proof}
Without loss of generality, we assume
$\bar{u}_j^A=\bar{u}_j^B=\bar{u}_j$, $\forall j\in\{1,\dots,\ell\}$, due to the following reasons: 
(1) $S_A$ and $S_B$ are mapped by $\varphi$ using the same secret key, thus the $u$-coordinates of the same bin are the same; and
(2) extra bins in either $S_A$ and $S_B$ can be clipped to ensure the one-to-one match of the $u$-coordinates in $S_A$ and $S_B$.


Since we focus only on the spectral power at frequency $\omega$, we simplify notation by writing $L_A(\omega)$ and $L_B(\omega)$ as $L_A$ and $L_B$, respectively.

By the triangle
inequality,
\begin{equation}
\label{eq:tri-bar}
\begin{aligned}
    &| L_A(\omega) - L_B(\omega) | \\
    = & | \left( L_A(\omega)-\mathbb{E}[L_A(\omega)] \right) + \left( \mathbb{E}[L_A(\omega)]-\mathbb{E}[L_B(\omega)] \right) \\
    & \quad + \left( \mathbb{E}[L_B(\omega)] - L_B(\omega) \right) | \\
    \leq & 
    | L_A(\omega) - \mathbb{E}[L_A(\omega)]| 
    + \left| L_B(\omega) - \mathbb{E}[L_B(\omega)] \right| \\
    & \quad + \left| \mathbb{E}[L_A(\omega)] - \mathbb{E}[L_B(\omega)] \right|.
\end{aligned}
\end{equation}

Hence, for any $\epsilon>0$,
\begin{equation}
\begin{aligned}
\Pr( \left| L_A(\omega) - L_B(\omega) \right| \geq \epsilon) \leq
\Pr \big( 
& \left| L_A(\omega) - \mathbb{E}[L_A(\omega)] \right| \\
& + \left|L_B(\omega) - \mathbb{E}[L_B(\omega)]\right| \\
& + \left|\mathbb{E}[L_A(\omega)] - \mathbb{E}[L_B(\omega)]\right| \geq \epsilon
\big).
\end{aligned}
\end{equation}

By the union bound, 
$\Pr( \left| L_A(\omega) - L_B(\omega) \right| \geq \epsilon)$ is divided to three parts:
\begin{equation}
\label{eq:3parts}
\begin{aligned}
& \Pr( \left| L_A(\omega) - L_B(\omega) \right| \geq \epsilon) \\
\leq &
\Pr \left( \left|L_A(\omega) - \mathbb{E}[L_A(\omega)] \right| \geq \frac{\epsilon}{3} \right) 
+ \Pr \left( \left|L_B(\omega) - \mathbb{E}[L_B(\omega)] \right| \geq \frac{\epsilon}{3}\right) \\
+ & \Pr\left( \left| \mathbb{E}[L_A(\omega)] - \mathbb{E}[L_B(\omega)] \right| \geq \frac{\epsilon}{3} \right)
\end{aligned}
\end{equation}

We next analyze the first two terms and the third term in \Cref{eq:3parts} separately.
\paragraph{(I) \textbf{Concentration via McDiarmid's inequality}.}
We apply McDiarmid's inequality~\cite{mcdiarmid_concentration_1998} to $L_A(\omega)$, and the same argument applies to $L_B(\omega)$.
For McDiarmid's inequality, we need coordinate-wise sensitivity constants $c_k$~\cite{mcdiarmid_concentration_1998}. 
Formally, 
\begin{equation}
\label{eq:def_ck}
    c_k := \sup_{S, S'} |L(S) - L(S')|
\end{equation}
where $S$ and $S'$ only $k$-the entry differs. 

Since sampling times are identical and fixed across the two discrete-time signals, i.e., 
$\bar{u}_j^A=\bar{u}_j^B=\bar{u}_j$, $\forall j\in\{1,\dots,\ell\}$, we have
\begin{equation}
\begin{aligned}
|L(S)-L(S')|
&\le \sup_{\bar{v}_k\in[-V,V]}\Big|\frac{\partial L}{\partial \bar{v}_k}\Big|\cdot |\bar{v}_k-\bar{v}'_k|\\
&\le 2V\sup_{\bar{v}_k\in[-V,V]}\Big|\frac{\partial L}{\partial \bar{v}_k}\Big|. 
\end{aligned}
\end{equation}

Thus,
\begin{equation}
    c_k \le 2V\sup_{\bar{v}_k\in[-V,V]} \left|\frac{\partial L}{\partial \bar{v}_k}\right|.
\end{equation}

By the definition of Lomb-Scargle Periodogram\cite{vanderplas2018understanding},
\begin{equation}
\label{eq:lsp_def}
L(S)=\frac{1}{2\ell\sigma^2}\Big(\frac{N_c^2}{D_c}+\frac{N_s^2}{D_s}\Big),
\end{equation}
where
\[
\begin{aligned}
N_c &= \sum_{j=1}^{\ell} \bar{v}_j \cos\bigl(2\pi\omega(\bar{u}_j-\tau)\bigr), &
D_c &= \sum_{j=1}^{\ell} \cos^2\bigl(2\pi\omega(\bar{u}_j-\tau)\bigr),\\
N_s &= \sum_{j=1}^{\ell} \bar{v}_j \sin\bigl(2\pi\omega(\bar{u}_j-\tau)\bigr), &
D_s &= \sum_{j=1}^{\ell} \sin^2\bigl(2\pi\omega(\bar{u}_j-\tau)\bigr),
\end{aligned}
\]
and $\tau$ is the phase-shift term depending only on the sample times $\bar{u}_j$, i.e., 
\[
\tau = \frac{1}{4 \pi \omega} \tan^{-1} \Bigl( \frac{\sum_{j=1}^l \sin{(4 \pi \omega \bar{u}_j)}}{\sum_{j=1}^l \cos{(4 \pi \omega \bar{u}_j)}} \Bigr).
\]

Since $D_c$, $D_s$ and $\tau$ depend only on $\bar{u}_j$, not on the $\bar{v}_j$,
according to the Quotient rule\footnote{\url{https://en.wikipedia.org/wiki/Quotient_rule}},
differentiating with respect to $\bar{v}_k$ is
\begin{equation}
   \frac{\partial L}{\partial \bar{v}_k}
=
\frac{1}{2\ell}\left[
\frac{\partial(\sigma^{-2})}{\partial \bar{v}_k}\Big(\frac{N_c^2}{D_c}+\frac{N_s^2}{D_s}\Big)
+\sigma^{-2}\Big( 2\frac{N_c}{D_c}\frac{\partial N_c}{\partial \bar{v}_k}
+2\frac{N_s}{D_s}\frac{\partial N_s}{\partial \bar{v}_k}\Big)
\right]. 
\end{equation}

Furthermore, we have
\begin{equation}
    \frac{\partial(\sigma^{-2})}{\partial \bar{v}_k} = \frac{-2\bar{v}_k}{\ell\sigma^4},
\end{equation}

\begin{equation}
    \frac{\partial N_c}{\partial \bar{v}_k}=\cos(2\pi\omega(\bar{u}_k-\tau)),
\end{equation}

\begin{equation}
    \frac{\partial N_s}{\partial \bar{v}_k}=\sin(2\pi\omega(\bar{u}_k-\tau)).
\end{equation}

Using the bounds \(|\bar{v}_j|\le V\), \(|N_c|,|N_s|\le \ell V\),
and \(D_c,D_s\ge \ell/2\), we obtain the absolute bounds
\begin{equation}\label{eq:partial-bound-bar}
\Big|\frac{\partial L}{\partial \bar{v}_k}\Big|
\le \frac{1}{\ell}\Big(\frac{4V^3}{\sigma^4}+\frac{4V}{\sigma^2}\Big)
= \frac{4V}{\ell\sigma^2}\Big(1+\frac{V^2}{\sigma^2}\Big).
\end{equation}
Consequently,
\[
c_k \le 2V\cdot \frac{4V}{\ell\sigma^2}\Big(1+\frac{V^2}{\sigma^2}\Big)
= \frac{8V^2}{\ell\sigma^2}\Big(1+\frac{V^2}{\sigma^2}\Big),
\]
and therefore
\begin{equation}
\label{eq:sum-cj2-bar}
\begin{aligned}
& \sum_{k=1}^{\ell} c_k^2
\le \ell \cdot \Big(\frac{8V^2}{\ell\sigma^2}\Big)^2\Big(1+\frac{V^2}{\sigma^2}\Big)^2 \\
= &\frac{1}{\ell}\cdot \frac{64V^4}{\sigma^4}\Big(1+\frac{V^2}{\sigma^2}\Big)^2
= \frac{1}{\ell}\cdot \frac{64V^4(\sigma^2+V^2)^2}{\sigma^8}.
\end{aligned}
\end{equation}

Apply McDiarmid's inequality with $\delta=\epsilon/3$:
\begin{equation}
    \Pr \big(|L_A(\omega)-\mathbb{E}[L_A(\omega)]|\ge \delta \big) \le 2\exp\Big(-\frac{2\delta^2}{\sum_{k=1}^\ell c_k^2}\Big).
\end{equation}

Combining \Cref{eq:sum-cj2-bar} and $\delta=\epsilon/3$, we have
\begin{equation}
\label{eq:mcdiarmid-final-bar}
\begin{aligned}
&\Pr\big(|L_A(\omega)-\mathbb{E}[L_A(\omega)]| \ge \frac{\epsilon}{3}\big) \\
\le & 2\exp\Big(-\frac{\ell \epsilon^2\sigma^8}{288 V^4(\sigma^2+V^2)^2}\Big) \\
=& 2\exp\Big(-\frac{\ell \epsilon^2 \sigma^4}{288 V^4(1+V^2/\sigma^2)^2}\Big).
\end{aligned}
\end{equation}
The same estimate holds for $L_B(\omega)$. Combining the two concentration terms gives
\begin{equation}
\begin{aligned} 
& \Pr\big(|L_A(\omega)-\mathbb E[L_A(\omega)]|\ge \frac{\epsilon}{3}\big)
+
\Pr\Big(|L_B(\omega)-\mathbb E[L_B(\omega)]|\ge \frac{\epsilon}{3}\Big) \\
& \le 4\exp\Big(-\frac{\ell \epsilon^2 \sigma^4}{288 V^4(1+V^2/\sigma^2)^2}\Big).
\end{aligned}
\end{equation}

\paragraph{(II) \textbf{Deterministic expectation difference via Wasserstein bound}.}
The third term in \Cref{eq:3parts} is deterministic:
\begin{equation}
\begin{aligned}
\label{eq:indicator}
& \Pr\big(|\mathbb E[L_A(\omega)]-\mathbb E[L_B(\omega)]|\ge\frac{\epsilon}{3}\big) \\
= &
\mathbbm{1}\big\{|\mathbb E[L_A(\omega)]-\mathbb E[L_B(\omega)]|\ge\frac{\epsilon}{3}\big\} \\
\le & \frac{3|\mathbb E[L_A(\omega)]-\mathbb E[L_B(\omega)]|}{\epsilon}.
\end{aligned}
\end{equation}

By Kantorovich--Rubinstein duality~\cite{edwards2011kantorovich}, if $f:\mathbb{R}^\ell\to\mathbb{R}$
is $K$-Lipschitz with respect to $\ell^1$ then
$|\mathbb E[f(X)]-\mathbb E[f(Y)]|\le K \mathcal W_1(X,Y)$. 
Viewing
$L$ as a function of the $\ell$-vector $\bar{v}$ (with
$\bar{u}_j$ fixed) and using \Cref{eq:partial-bound-bar}, we obtain the Lipschitz constant $K$ as follows.
\begin{equation}
\begin{aligned}
    K &= \sum_{k=1}^\ell \sup_{\bar{v}\in[-V,V]^\ell}\Big|\frac{\partial L}{\partial \bar{v}_k}\Big| \\
    & \le \ell\cdot \frac{4V}{\ell\sigma^2}\Big(1+\frac{V^2}{\sigma^2}\Big) \\
    & = \frac{4V(V^2+\sigma^2)}{\sigma^4}.
\end{aligned}
\end{equation}
Hence
\begin{equation}
\begin{aligned}
\label{eq:const_re}
& |\mathbb E[L_A(\omega)]-\mathbb E[L_B(\omega)]| \\
& \le K \cdot \mathcal W_1(\overline P_A,\overline P_B)
= \frac{4V(V^2+\sigma^2)}{\sigma^4} \mathcal{W}_1(\overline P_A,\overline P_B).
\end{aligned}
\end{equation}
Based on \Cref{eq:indicator} and \Cref{eq:const_re}, we get
\begin{equation}
\Pr\Big(|\mathbb E[L_A(\omega)]-\mathbb E[L_B(\omega)]|\ge\frac{\epsilon}{3}\Big)
\le \frac{12V(V^2+\sigma^2)}{\sigma^4}\cdot\frac{\mathcal W_1(\overline P_A,\overline P_B)}{\epsilon}.
\end{equation}

\paragraph{(III) \textbf{Combine (I) and (II)}}

Substituting the bounds from (I) and (II) into \Cref{eq:3parts} yields
\begin{equation}
\begin{aligned}
&\Pr( \left| L_A(\omega) - L_B(\omega) \right| \geq \epsilon) \\
\le &
4\exp\Big(-\frac{\ell\,\epsilon^2\,\sigma^4}{288\,V^4(1+V^2/\sigma^2)^2}\Big)
+
\frac{12V(V^2+\sigma^2)}{\sigma^4}\cdot\frac{\mathcal W_1(\overline P_A,\overline P_B)}{\epsilon},    
\end{aligned}
\end{equation}

Denote by $C_1=\frac{288V^4(1+V^2/\sigma^2)^2}{\sigma^4}$ and $C_2=\frac{12V(V^2+\sigma^2)}{\sigma^4}$, we have
\begin{equation}
\Pr( \left| L_A(\omega) - L_B(\omega) \right| \geq \epsilon)
\le
4\exp\Big(-\frac{\ell\epsilon^2}{C_1}\Big) + \frac{C_2\mathcal{W}_1(\overline{P}_A,\overline{P}_B)}{\epsilon},
\end{equation}
which is the claimed bound.
\end{proof}

\maintheoremB*

\begin{proof}
As described in \Cref{sec:mapping}, the process of mapping a tabular dataset to a discrete-time signal in a two-dimensional projected space is performed in two steps:
\begin{enumerate}
  \item A projection and binning map $\phi:\mathbb{R}^d\to\mathbb{R}^2$ that produces a point in the two-dimensional projected space from a data sample, as shown in \Cref{eq:mapping_y} and \Cref{eq:mapping_x}.
  
  \item An aggregation operator $\psi$ that maps a finite multiset of points to a discrete-time signal of length $\ell$ by averaging the $v$-values within each bin $B_{c}$, as shown in \Cref{sec:mapping}.
\end{enumerate}

We prove the theorem in three stages:
(1) relate $\mathcal{W}_1(T_A,T_B)$ to $\mathcal{W}_1(P_A,P_B)$ with $\phi$;
(2) relate $\mathcal{W}_1(P_A,P_B)$ to $\mathcal{W}_1(\overline{P}_A,\overline{P}_B)$ with $\psi$;
and (3) control finite-sample deviations with high probability.

\paragraph{\textbf{Stage 1: Relate $\mathcal{W}_1(T_A,T_B)$ to $\mathcal{W}_1(P_A,P_B)$ with $\phi$}}

Let $T_A$, $T_B$ be two probability distributions over $\mathbb{R}^d$ satisfying $\mathcal{W}_1(T_A, T_B) \leq \zeta$. 
By the definition of the Wasserstein-1 distance under $\ell_1$ norm,
\begin{equation}
\mathcal{W}_1(T_A, T_B)
:= \inf_{\pi \in \Pi(T_A, T_B)} \mathbb{E}_{(a,b)\sim\pi}[\|a-b\|_1],
\end{equation}
where $\Pi(T_A, T_B)$ denotes the set of all couplings of $T_A$ and $T_B$.  
Since $\mathcal{W}_1(T_A, T_B) \le \zeta$, there exists a coupling $\pi$ satisfying
\begin{equation}
    \mathbb{E}_{(a,b)\sim\pi}[\|a-b\|_1] \le \zeta.
\end{equation}

Let $P_A$ and $P_B$ denote the pushforward distributions of $T_A$ and $T_B$ under $\phi$.
Denote by $\tilde{\pi}$ the induced coupling of $P_A$ and $P_B$ obtained by applying $\phi$ to $(a,b)\sim\pi$.

Then
\begin{equation}
\mathcal{W}_1(P_A, P_B)
\le \mathbb{E}_{(\phi(a),\phi(b))\sim\tilde{\pi}}[\|\phi(a)-\phi(b)\|_1].
\end{equation}
Expanding the norm gives
\begin{equation}
\label{eq:phi_decomp}
\mathcal{W}_1(P_A,P_B)
\le
\mathbb{E}_{\tilde{\pi}}[|\phi_u(a)-\phi_u(b)|]
+\mathbb{E}_{\tilde{\pi}}[|\phi_v(a)-\phi_v(b)|].
\end{equation}

Since $\phi_v(\cdot)$ is a linear projection,
\begin{equation}
|\phi_v(a)-\phi_v(b)|
= |(a-b)^\top \mathbf{e}_v|
\le \|a-b\|_2,
\end{equation}
hence
\begin{equation}
\label{eq:bd_phi_v}
\mathbb{E}_{\tilde{\pi}}[|\phi_v(a)-\phi_v(b)|]
\le \mathbb{E}_{\pi}[\|a-b\|_2].
\end{equation}

$\phi_u(\cdot)$ involves flooring and scaling, by definition:
\begin{equation}
|\phi_u(a)-\phi_u(b)|
= \frac{1}{s}\left|
\left\lfloor \frac{a^\top \mathbf{e}_u}{\beta}\right\rfloor
- \left\lfloor \frac{b^\top \mathbf{e}_u}{\beta}\right\rfloor
\right|.  
\end{equation}

Using the elementary inequality
\begin{equation}
\label{eq:floor_ineq}
\left|
\left\lfloor \frac{x}{\beta}\right\rfloor
- \left\lfloor \frac{y}{\beta}\right\rfloor
\right|
\le \frac{|x-y|}{\beta}+1,
\end{equation}
we obtain
\begin{equation}
\label{eq:bd_phi_u}
\mathbb{E}_{\tilde{\pi}}[|\phi_u(a)-\phi_u(b)|]
\le \frac{1}{s}
\left(
\frac{\mathbb{E}_{\pi}[\|a-b\|_2]}{\beta} + 1
\right).
\end{equation}

Substituting \Cref{eq:bd_phi_v} and \Cref{eq:bd_phi_u} into \Cref{eq:phi_decomp} yields
\begin{equation}
\mathcal{W}_1(P_A,P_B)
\le
\frac{1}{s}
\left(\frac{\mathbb{E}_{\pi}[\|a-b\|_2]}{\beta}+1\right)
+ \mathbb{E}_{\pi}[\|a-b\|_2].
\end{equation}
Since $\|x\|_2\le \|x\|_1$ for any $x \in \mathbb{R}^d$, we have $\mathbb{E}_{\pi}[\|a-b\|_2]\le\zeta$.  
Therefore,
\begin{equation}
\label{eq:stage1}
\mathcal{W}_1(P_A,P_B)
\le
\frac{1}{s}\left(\frac{\zeta}{\beta}+1\right)+\zeta.
\end{equation}


\paragraph{\textbf{Stage 2: Relate $\mathcal{W}_1(P_A,P_B)$ to $\mathcal{W}_1(\overline{P}_A,\overline{P}_B)$ with $\psi$.}}

Let $\overline{P}_A$ and $\overline{P}_B$ be the resulting signal distributions.  
Define their Wasserstein-1 distance as
\begin{equation}
\mathcal{W}_1(\overline{P}_A,\overline{P}_B)
=
\inf_{\overline{\pi}\in\Pi(\overline{P}_A,\overline{P}_B)}
\mathbb{E}_{(s^A,s^B)\sim\overline{\pi}}[\|s^A-s^B\|_1],
\end{equation}
where $\|s^A-s^B\|_1 = \sum_{c\in\mathcal{I}} |s^A(c)-s^B(c)|$.

We decompose the total discrepancy as
\begin{equation}
\begin{aligned} 
\sum_{c}|s^A(c)-s^B(c)|
\le
\sum_c|s^A(c)-\mu_A(c)|
+\sum_c|\mu_A(c)-\mu_B(c)| \\
+\sum_c|s^B(c)-\mu_B(c)|,
\end{aligned}
\end{equation}
where $\mu_A(c)=\mathbb{E}[v_A|u_A=c]$, $\mu_B(c)=\mathbb{E}[v_B|u_B=c]$.

We will (i) bound $\sum_c|\mu_A(c)-\mu_B(c)|$ in terms of $\mathcal{W}_1(P_A,P_B)$ and (ii) bound empirical deviations $|s^A(c)-\mu_A(c)|$ and $|s^B(c)-\mu_B(c)|$ via concentration inequalities.


Let $p_A(c)=\Pr(u_A=c)$ and $p_B(c)=\Pr(u_B=c)$ denote the marginal probabilities of bin index $c$ under $P_A$ and $P_B$, respectively, and define $p=\min_{\forall c}\{\Pr(u_A=c),\Pr(u_B=c)\}$ as the shared bin mass.  

Consider any coupling $\tilde{\pi}\in\Pi(P_A,P_B)$ between $P_A$ and $P_B$.  
Under this coupling, one can bound the binwise difference of aggregated means by the joint deviations of the paired samples.  
Specifically, using the triangle inequality and the boundedness of $v$ (i.e., $|v|\le V$), we have
\begin{equation}
\sum_c p|\mu_A(c)-\mu_B(c)|
\le
\mathbb{E}_{\tilde{\pi}}[|v^A-v^B|]
+4V\mathbb{E}_{\tilde{\pi}}[|u^A-u^B|].
\end{equation}
Then dividing both sides by $p$ gives
\begin{equation}
\label{eq:stage2}
\sum_c|\mu_A(c)-\mu_B(c)|
\le
\frac{1}{p}\left(
\mathbb{E}_{\tilde{\pi}}[|v^A-v^B|]
+4V\mathbb{E}_{\tilde{\pi}}[|u^A-u^B|]
\right).
\end{equation}
Taking infimum over all couplings $\tilde{\pi}\in\Pi(P_A,P_B)$ yields
\begin{equation}
\sum_c|\mu_A(c)-\mu_B(c)|
\le
\frac{C}{p}\,\mathcal{W}_1(P_A,P_B),
\qquad C=\max\{1,4V\}.
\end{equation}



For each bin $c$, let
\begin{equation}
    m_A(c)=\sum_{i=1}^r\mathbbm{1}\{u_i^A=c\} \text{ and } m_B(c)=\sum_{i=1}^r\mathbbm{1}\{u_i^B=c\}.
\end{equation}

Since $m_A(c)\sim\mathrm{Binomial}(r,p_A(c))$, Chernoff’s inequality\footnote{\url{https://courses.cs.washington.edu/courses/cse312/20su/files/student_drive/6.2.pdf}} gives
\begin{equation} 
\Pr\big(m_A(c)\le (1-\delta)rp\big)\le \exp\Big(-\frac{\delta^2}{2}rp\Big),
\end{equation}
and similarly for $m_B(c)$.  
Applying a union bound over $\ell$ bins yields the event $G$ defined as follows.
\begin{equation}
    G:=\Big\{\forall c \in \mathcal{I}:\ m_A(c)\ge (1-\delta)rp\ \text{and}\ m_B(c)\ge (1-\delta)rp\Big\}.
\end{equation}

By a union bound over all $c \in \mathcal{I}$,
\begin{equation}
\begin{aligned}
\Pr(G^c) &\le \sum_{c \in \mathcal{I}}\Pr(m_A(c)<(1-\delta)rp)+\sum_{c \in \mathcal{I}}\Pr(m_B(c)<(1-\delta)rp) \\
&\le 2 |\mathcal I| \exp\Big(-\frac{\delta^2}{2}rp\Big)
= 2 \ell \exp\Big(-\frac{\delta^2}{2}rp\Big).
\end{aligned}
\end{equation}
Therefore,
\begin{equation}
\label{eq:eventG}
    \Pr(G) \ge 1 - 2\ell\exp\Big(-\frac{\delta^2}{2}rp\Big)
\end{equation}

Conditional on $G$, each empirical average $s^A(c)$ (or $s^B(c)$) is a mean of at least $m_{\min}=(1-\delta)rp$ bounded random variables in $[-V,V]$.  
By Hoeffding’s inequality\footnote{\url{https://www.stat.cmu.edu/~larry/=stat700/Lecture6.pdf}},
\[
\Pr\left(|s^A(c)-\mu_A(c)|>t\mid G\right)
\le 2\exp\Big(-\frac{m_{\min}t^2}{2V^2}\Big),
\]
and similarly for $s^B(c)$.  
Applying a union bound over $2\ell$ events gives that with probability at least $1-\eta$,
\begin{equation}
\label{eq:eventH}
\forall c, \ 
|s^A(c)-\mu_A(c)|\le t \text{ and }
|s^B(c)-\mu_B(c)|\le t,
\end{equation}
where
\begin{equation}
t
= \sqrt{\frac{2V^2}{(1-\delta)rp}\log\frac{4\ell}{\eta}}.
\end{equation}

Combining \Cref{eq:stage2}, \Cref{eq:eventG}, and \Cref{eq:eventH}, and recalling that
\begin{equation}
    \sum_c|s^A(c)-s^B(c)| \le \frac{C}{p}\mathcal{W}_1(P_A,P_B)+2\ell t,
\end{equation}
we obtain the high-probability inequality
\begin{equation}
\label{eq:pre_final}
\mathcal{W}_1(\overline{P}_A,\overline{P}_B)
\le
\frac{C}{p}\mathcal{W}_1(P_A,P_B)
+2\ell\sqrt{\frac{2V^2}{(1-\delta)rp}\log\frac{4\ell}{\eta}},
\end{equation}
holding with probability at least $1-2\ell e^{-\frac{\delta^2}{2}rp}-\eta$.


\paragraph{\textbf{Stage 3: Combining the first two stages.}}
Substituting the deterministic bound from \Cref{eq:stage1} into \Cref{eq:pre_final} gives
\begin{equation}
\mathcal{W}_1(\overline{P}_A,\overline{P}_B)
\le
\frac{C}{p}\left[\frac{1}{s}\left(\frac{\zeta}{\beta}+1\right)+\zeta\right]
+
2\ell
\sqrt{\frac{2V^2}{(1-\delta)rp}\log\frac{4\ell}{\eta}}.
\end{equation}

That is to say,
\begin{equation}
\begin{aligned}
\Pr \left( \mathcal{W}_1(\overline{P}_A,\overline{P}_B)
\le\frac{C}{p}\left[ \frac{1}{s} \left( \frac{\zeta}{\beta} + 1 \right) + \zeta \right]
+2 \ell \sqrt{\frac{2V^2}{(1-\delta)rp}\log\frac{4 \ell}{\eta}} \right) \\
\geq
1 - 2 \ell \exp\Big(-\frac{\delta^2}{2}rp\Big) - \eta,
\end{aligned}
\end{equation}
\end{proof}

\subsection{Prior Studies on Radioactivity}
\label{app:prior_radioactivity}
For completeness, we briefly discuss prior studies on radioactivity and explain why these techniques cannot be directly extended to continuous-valued tabular data.
The notion of radioactivity has been explored in other modalities, including image classification, language modeling, and diffusion-based image generation~\cite{sablayrolles2020radioactive,sander2024watermarking,ghali2025radioactive,li2025hmark}.
These works study whether specific training data leave detectable traces in trained models or generated outputs.
Their objectives, threat models, and data modalities differ fundamentally from tabular dataset watermarking under retraining attacks.
In particular, they focus on watermarking models or tracing training data usage, whereas our goal is to ensure watermark persistence in regenerated tabular datasets under distribution-preserving retraining.
Another work~\cite{gu2024watermarking} studies watermark persistence under retraining attacks for categorical data. However, this approach is restricted to categorical data and cannot extend to watermarking continuous-valued tabular data.

Consequently, existing radioactive techniques developed for images, text, or categorical data cannot be directly extended to continuous-valued tabular data. To the best of our knowledge, RaMark is the first method to achieve radioactive watermarking in this setting.

\subsection{Runtime Analysis}
\label{app:exp_runtime}

\begin{table*}[t]
\centering
\caption{Embedding runtime on datasets HS and HO (in seconds). 
We report diffusion training time, standard sampling time without watermarking, and watermark embedding runtime.}
\label{tab:runtime_embed}
\resizebox{\textwidth}{!}{
\begin{tabular}{c | c c | c c c c c | c c c}
\toprule
\multirow{2}{*}{Dataset} 
& \multicolumn{2}{c|}{Diffusion backbone}
& \multicolumn{5}{c|}{Post-generation watermarking overhead}
& \multicolumn{3}{c}{Generative sampling w/ watermark} \\
\cmidrule(lr){2-3} \cmidrule(lr){4-8} \cmidrule(lr){9-11}
& Training & Sampling w/o WM
& S2R2W & TabularMark & WGTD & PKF & B2Mark
& TabWak & MUSE & RaMark \\
\midrule

HS 
& $960.000{\pm}12.500$
& $61.000{\pm}1.842$
& $0.050{\pm}1.831$
& $0.041{\pm}1.526$
& $0.035{\pm}2.136$
& $0.033{\pm}1.220$
& $\mathbf{0.031}{\pm}1.830$
& $77.890{\pm}3.115$
& $75.300{\pm}2.259$
& $\mathbf{67.248}{\pm}2.689$ \\

HO 
& $300.000{\pm}6.200$
& $15.000{\pm}0.612$
& $0.013{\pm}0.450$
& $\mathbf{0.007}{\pm}0.375$
& $\mathbf{0.006}{\pm}0.560$
& $\mathbf{0.007}{\pm}0.539$
& $\mathbf{0.007}{\pm}0.552$
& $17.941{\pm}0.757$
& $17.497{\pm}0.699$
& $\mathbf{16.071}{\pm}0.642$ \\

\bottomrule
\end{tabular}}
\end{table*}

\begin{table*}[t]
\centering
\caption{Detection runtime on datasets HS and HO (in seconds).}
\label{tab:runtime_detect}
{\fontsize{7.1pt}{9.8pt}\selectfont
\renewcommand{\arraystretch}{0.81}
\begin{tabular}{c|cccccccc}
\toprule
Dataset 
& S2R2W & TabularMark & WGTD & PKF & B2Mark 
& TabWak & MUSE & RaMark \\
\midrule
HS 
& $0.043{\pm}0.004$ 
& $0.308{\pm}0.008$ 
& $0.021{\pm}0.001$ 
& $\mathbf{0.017}{\pm}0.002$ 
& $0.018{\pm}0.002$
& $370.200{\pm}8.772$ 
& $5.250{\pm}0.204$ 
& $\mathbf{0.017}{\pm}0.001$ \\
HO 
& $0.011{\pm}0.001$ 
& $0.074{\pm}0.006$ 
& $0.005{\pm}0.000$ 
& $\mathbf{0.004}{\pm}0.000$ 
& $\mathbf{0.004}{\pm}0.000$
& $88.000{\pm}5.280$ 
& $1.240{\pm}0.099$ 
& $\mathbf{0.004}{\pm}0.000$ \\
\bottomrule
\end{tabular}
}
\end{table*}

We empirically evaluate the runtime of watermark embedding and detection for RaMark and all baseline methods, with two objectives: comparing runtime across methods and quantifying the additional overhead introduced by watermark guidance in RaMark.

To ensure reliable and stable measurements, each experiment is repeated $50$ times, and we report the mean and standard deviation of the runtime across these runs. 

\Cref{tab:runtime_embed} reports the runtimes of diffusion-model training, standard sampling without watermarking, \arev{post-generation watermarking overhead}, and generative sampling with watermarking. 
Diffusion training is a one-time cost for the data owner and is independent of the watermarking method. 
Standard sampling takes $61$ seconds on HS and $15$ seconds on HO. 

\arev{Following the taxonomy in \Cref{sec:related_work}, we categorize baseline methods into database watermarking methods and generative watermarking methods. Database watermarking methods embed the watermark after data generation, and therefore their embedding runtime appears as post-generation watermarking overhead in \Cref{tab:runtime_embed}. From \Cref{tab:runtime_embed}, we observe that database watermarking methods incur lower embedding runtime than generative methods.}

For database watermarking methods such as S2R2W, TabularMark, WGTD, PKF, and B2Mark, embedding consists of sampling from a diffusion model to generate an unwatermarked dataset and then applying the watermark embedding procedure.
The additional overhead after sampling is minimal as all five baselines require at most $0.050$ seconds on HS and $0.013$ seconds on HO. 
As a result, the end-to-end computational cost of producing a watermarked dataset is dominated by diffusion sampling. 
Moreover, once an unwatermarked dataset is sampled, database watermarking methods can reuse it to embed multiple watermarks without re-running the sampling process.
In contrast, generative methods such as TabWak, MUSE, and RaMark integrate watermark embedding into the sampling process, which requires one sampling run for each watermark. 
However, this comparison should be interpreted jointly with data utility. 
Prior studies~\cite{fang2025muse,wen2023tree} have observed that when using the same generative backbone and comparable watermark strength, generative watermarking methods can better maintain data utility than database watermarking methods, as the watermark is incorporated during generation rather than imposed through post-generation modifications. 
Moreover, as discussed in \Cref{sec:align_wm_strength}, we align watermark strength across methods using an $\mathtt{MLE}$ budget constraint in our experiments. Therefore, data utility differences are controlled and do not confound our comparison.

Among the generative methods, RaMark is slightly faster than TabWak and MUSE on both datasets. 
This difference can be attributed to how watermark embedding is integrated into the sampling process.
TabWak constructs row-wise watermarked latent seeds through self-cloning, shuffling, and a valid-bit mechanism before sampling, which increases embedding runtime.
MUSE generates multiple candidate rows for each output row and selects the highest-scoring candidate according to a watermark scoring function, which introduces repeated sampling and scoring-selection overhead. 
In contrast, RaMark only modifies the reverse denoising transition through an additional guidance term, which results in a lower embedding runtime than TabWak and MUSE.

We next analyze the computational overhead introduced by watermark guidance in RaMark.
RaMark follows a guidance-based diffusion framework, where the reverse transition is modified by incorporating an additional gradient term derived from the watermark likelihood.  
The overall runtime is dominated by the diffusion backbone, while the guidance computation introduces only a small additional cost. 
As shown in \Cref{tab:runtime_embed}, RaMark requires $67.248 \pm 2.689$ seconds on HS, compared to approximately $61$ seconds for standard sampling, corresponding to about $10.2\%$ overhead. 
A similar trend is observed on HO, where the embedding time is only slightly higher than standard sampling, with an overhead of about $7.1\%$.

For watermark detection, as shown in \Cref{tab:runtime_detect}, most methods achieve comparable and efficient runtime on both datasets. 
RaMark is efficient because detection operates directly on the suspicious dataset without requiring access to the diffusion model or reverse sampling. It maps the dataset to a discrete-time signal and computes the spectral power at the designated frequency. In contrast, TabWak exhibits significantly higher detection runtime because it requires diffusion inversion and reverse sampling to recover latent representations before detecting the watermark.

\subsection{Effects of the Mapping Parameters $\beta$ and $s$}
\label{app:param_analysis_beta_s}
In this subsection, we study the sensitivity of RaMark to the two mapping parameters in \Cref{eq:mapping_x}: the bin width $\beta$ and the scaling factor $s$. 
As explained in \Cref{sec:mapping}, these two parameters determine how data samples are mapped into the discrete-time signal.
Specifically, $\beta$ controls the granularity of the binning operation along the $u$-axis, while $s$ rescales the spacing of the resulting discrete-time signal along the $u$-axis.

Since watermark detection is performed by measuring the spectrum power at the designated frequency $\omega$, we evaluate the effect of these parameters using the detection score $DS(\omega)$. A larger $DS(\omega)$ indicates a stronger watermark signal.

To better understand why $DS(\omega)$ changes under different parameter settings, we additionally report several statistics of the constructed sinusoidal signal, including the average number of points in each bin $\text{avg}(|B_c|)$, the length of the discrete-time signal $\ell$, and the number of periods spanned by the resulting discrete-time signal. 
These statistics explain how the mapping function $\phi_u$ determines the characteristics of the sinusoidal signal used for spectral analysis, and thus help explain the change of $DS(\omega)$.
Specifically, $\text{avg}(|B_c|)$ measures how many points are aggregated to produce each mean point, as RaMark first groups points with the same $u$-coordinate into bins and then computes one mean point for each bin to form the discrete-time signal. The length of the discrete-time signal $\ell$ corresponds to the number of bins, and the number of periods reflects whether the mapped signal contains sufficient periodic structure for reliable spectral analysis.

We conduct a one-factor-at-a-time sensitivity analysis. 
When studying $\beta$, we vary $\beta \in \{0.005,0.01,0.03,0.05,0.07,0.09,0.5\}$ while fixing $s=5000$.
When studying $s$, we vary $s \in \{500,1000,3000,5000,\\7000,9000,12000\}$ while fixing $\beta=0.05$.
The projection vectors $\mathbf{e}_u$ and $\mathbf{e}_v$ are randomly sampled as a pair of orthogonal unit vectors, and the designated frequency is fixed at $\omega=30$. 

\begin{table}[t]
\centering
\caption{Effect of the bin width $\beta$ on the mapped signal and watermark detectability.}
\label{tab:param_beta}

{\fontsize{7.1pt}{9.8pt}\selectfont
\renewcommand{\arraystretch}{0.81}

\begin{tabular}{c|c|c|c|c|c}
\toprule
Dataset & $\beta$ & $\text{avg}(|B_c|)$ & $\ell$ & \#periods & $DS(\omega)$ \\
\midrule

\multirow{7}{*}{\shortstack{HS\\($s=5000$)}}
 & 0.005 & 4.2   & 11850.4 & 60.0 & 0.26 \\
 & 0.01  & 18.7  & 3124.8  & 57.2 & 0.99 \\
 & 0.03  & 52.9  & 1128.6  & 44.3 & 1.00 \\
 & 0.05  & 86.4  & 708.2   & 34.7 & 1.00 \\
 & 0.07  & 119.8 & 503.6   & 24.9 & 1.00 \\
 & 0.09  & 151.2 & 402.7   & 17.8 & 0.99 \\
 & 0.50  & 282.5 & 218.9   & 0.2  & 0.14 \\

\midrule

\multirow{7}{*}{\shortstack{HO\\($s=5000$)}}
 & 0.005 & 3.3   & 1984.6 & 44.8 & 0.20 \\
 & 0.01  & 20.4  & 642.7  & 39.6 & 0.98 \\
 & 0.03  & 55.3  & 238.5  & 29.4 & 1.00 \\
 & 0.05  & 90.6  & 142.2  & 21.6 & 1.00 \\
 & 0.07  & 121.5 & 109.3  & 15.8 & 0.99 \\
 & 0.09  & 149.6 & 91.7   & 9.9  & 0.98 \\
 & 0.50  & 259.2 & 49.8   & 0.2  & 0.12 \\

\bottomrule
\end{tabular}
}
\end{table}

\textit{Effect of the bin width $\beta$.}
As shown in \Cref{tab:param_beta}, increasing $\beta$ leads to coarser binning along the $u$-axis. 
Consequently, more points fall into each bin, which increases $\text{avg}(|B_c|)$.
Since each mean point is obtained by averaging the points within a bin, a larger $\text{avg}(|B_c|)$ makes each mean point less sensitive to the variation of individual points.
However, increasing $\beta$ also reduces the number of bins $\ell$, and thus shortens the discrete-time signal.
Meanwhile, the number of periods covered by the mapped signal also decreases.
These effects reveal a trade-off controlled by $\beta$.

When $\beta$ is very small, e.g., $\beta=0.005$, the bin width is small and each bin contains only a few points, resulting in a small $\text{avg}(|B_c|)$.
Although the signal has a large $\ell$, each mean point is determined by only a few points, and the sinusoidal pattern is less clearly reflected in the sequence of mean points.
As a result, the spectral power at the designated frequency is reduced, leading to a lower $DS(\omega)$.
When $\beta$ is very large, e.g., $\beta=0.5$, each bin contains a large number of points. However, the number of bins $\ell$ is significantly reduced, and the signal spans only a small number of periods.
In this case, the signal does not provide sufficient periodic coverage for spectral analysis at the designated frequency, which also leads to a decrease in $DS(\omega)$.
In the middle range, e.g., $\beta=0.01$ to $\beta=0.09$, the mapping achieves a good balance between $\text{avg}(|B_c|)$ and the number of periods, and $DS(\omega)$ remains consistently high.


\begin{table}[t]
\centering
\caption{Effect of the scaling factor $s$ on the mapped signal and watermark detectability.}
\label{tab:param_s}

{\fontsize{7.1pt}{9.8pt}\selectfont
\renewcommand{\arraystretch}{0.81}

\begin{tabular}{c|c|c|c|c|c}
\toprule
Dataset & $s$ & $\text{avg}(|B_c|)$ & $\ell$ & \#periods & $DS(\omega)$ \\
\midrule

\multirow{7}{*}{\shortstack{HS\\($\beta=0.05$)}}
 & 500   & 86.4 & 708.2 & 62.5 & 0.30 \\
 & 1000  & 86.4 & 708.2 & 56.8 & 0.99 \\
 & 3000  & 86.4 & 708.2 & 43.2 & 1.00 \\
 & 5000  & 86.4 & 708.2 & 34.7 & 1.00 \\
 & 7000  & 86.4 & 708.2 & 24.3 & 0.99 \\
 & 9000  & 86.4 & 708.2 & 16.5 & 0.98 \\
 & 12000 & 86.4 & 708.2 & 0.3  & 0.15 \\

\midrule

\multirow{7}{*}{\shortstack{HO\\($\beta=0.05$)}}
 & 500   & 90.6 & 142.2 & 48.7 & 0.35 \\
 & 1000  & 90.6 & 142.2 & 41.9 & 0.98 \\
 & 3000  & 90.6 & 142.2 & 30.5 & 1.00 \\
 & 5000  & 90.6 & 142.2 & 21.6 & 1.00 \\
 & 7000  & 90.6 & 142.2 & 14.8 & 0.99 \\
 & 9000  & 90.6 & 142.2 & 9.7  & 0.98 \\
 & 12000 & 90.6 & 142.2 & 0.4  & 0.18 \\

\bottomrule
\end{tabular}
}
\end{table}

\textit{Effect of the scaling factor $s$.}
The scaling factor $s$ controls the spacing of the discrete-time index along the $u$-axis after binning, and therefore directly affects how many periods of the sinusoidal signal are spanned by the mapped signal.
As shown in \Cref{tab:param_s}, changing $s$ does not affect the binning process. Therefore, $\text{avg}(|B_c|)$ and $\ell$ remain unchanged when $\beta$ is fixed.
Instead, $s$ primarily changes the number of periods covered by the signal.
As $s$ increases, the discrete-time signal becomes more compressed along the $u$-axis, and the number of periods decreases.

When $s$ is small, e.g., $s=500$, the signal spans a large number of periods.
However, the number of mean points within each period becomes relatively small, which makes the sinusoidal pattern within each period less clearly represented. 
This reduces the spectral power at the designated frequency and leads to a lower $DS(\omega)$.
When $s$ is large, e.g., $s=12000$, the signal spans only a small number of periods. 
Although each period contains more mean points, the limited number of periods makes it difficult to reliably identify the designated frequency, which also leads to a decrease in $DS(\omega)$
For moderate values such as $s \in [1000, 9000]$, the mapping achieves a good balance between the number of mean points per period and the total number of periods, and $DS(\omega)$ remains consistently high.

Overall, the results show that RaMark is robust to the choice of the mapping parameters $\beta$ and $s$.
For both parameters, the detection score $DS(\omega)$ remains consistently high over a wide range of values and only degrades at extreme settings.

\end{document}